\documentclass[a4paper,UKenglish,cleveref, autoref, thm-restate]{lipics-v2021}
\makeatletter
\def\@hideLIPIcs{1} %
\makeatother
\nolinenumbers 
\usepackage{hyperref}
\usepackage{tabularx}
\usepackage{rotating}
 
\usepackage{amsfonts} 
\usepackage{mathtools} 
\usepackage{graphicx} 
\usepackage{algorithmic} 
\usepackage{algorithm} 
\usepackage{ulem} 
\usepackage{multicol}
\usepackage[referable]{threeparttablex} 
\usepackage{tabu} 
\usepackage{multirow} 
\usepackage{makecell} 
\usepackage{longtable} 
\usepackage{booktabs}
\usepackage{pifont} 
\usepackage{tikz} 
\usetikzlibrary{arrows.meta,positioning, tikzmark,calc}
\usepackage{subcaption} 
\usepackage{wrapfig}

\usepackage{mdframed}
\newcommand{\mc}[3]{\multicolumn{#1}{#2}{#3}}

\usepackage{xcolor} 
\usepackage{colortbl} 
\usepackage{caption} 
\usepackage{calc} 
\usepackage{pdflscape} 
\usepackage{supertabular} 
\usepackage[autolanguage]{numprint} 
\npthousandsep{\,}
\usepackage{placeins}
\newcommand{\NP}{\ensuremath{\mathbf{NP}}}
\definecolor{green}{RGB}{0, 150, 130}
\definecolor{green70}{RGB}{76, 181, 167}
\definecolor{blue}{RGB}{70, 100, 170}
\definecolor{blue70}{RGB}{125,146,195}
\definecolor{blue50}{RGB}{162,177,212}                                                           
\definecolor{blue30}{RGB}{199,208,229}
\definecolor{blue15}{RGB}{227,232,242}
\definecolor{lightgray}{rgb}{0.86,0.86,0.86}
\definecolor{greengray}{RGB}{193,228,224}
\definecolor{redgray}{RGB}{233,183,183}
\definecolor{red}{RGB}{204,0,0}

\newcommand{\G}{\mathcal{G}}
\newcommand{\E}{\mathcal{E}}
\newcommand{\mS}{\mathcal{S}}
\newcommand{\al}{\beta}

\newcommand{\hpack}{\textsc{red2pack\_heuristic}}
\newcommand{\epack}{\textsc{red2pack\_b\&r}}
\newcommand{\redp}{\textsc{red2pack}}
\newcommand{\redpfull}{\textsc{elaborated}}
\newcommand{\redpdc}{\textsc{core}}
\newcommand{\pack}{\textsc{2pack}}
\newcommand{\nex}[1]{N[#1]}%
\newcommand{\n}[1]{N(#1)}%
\newcommand{\twone}[1]{N_2[#1]}%
\newcommand{\twon}[1]{N_2(#1)}%
\newcommand{\degtwo}{\textnormal{deg}^2}%
\newtheorem{@reduction}{Reduction}[section]
\newenvironment{reduction}[1][\unskip]{\begin{@reduction}[#1]\noindent\newline\noindent\hspace*{-.18cm}}{\end{@reduction}}

\newcommand{\ie}{i.e.\ }
\newcommand{\etal}{et~al.~}

\usepackage{todonotes}

\title{Scalable Algorithms for 2-Packing Sets on Arbitrary Graphs}

\author{Jannick Borowitz}{Heidelberg University, Heidelberg, Germany}{jannick.borowitz@uni-heidelberg.de}{https://orcid.org/0000-0002-8419-6324}{}
\author{Ernestine Großmann}{Heidelberg University, Heidelberg, Germany}{e.grossmann@informatik.uni-heidelberg.de}{https://orcid.org/0000-0002-9678-0253}{}
\author{Christian Schulz}{Heidelberg University, Heidelberg, Germany}{christian.schulz@informatik.uni-heidelberg.de}{https://orcid.org/0000-0002-2823-3506}{}
\author{Dominik Schweisgut}{Heidelberg University, Heidelberg, Germany}{dominik.schweisgut@stud.uni-heidelberg.de}{https://orcid.org/0009-0008-2627-505X}{}

\authorrunning{J. Borowitz \etal} 

\Copyright{Jannick, Borowitz, Ernestine Großmann, Christian Schulz, Dominik Schweisgut} 



\category{} %

\relatedversion{} %

\acknowledgements{We acknowledge support by DFG grant SCHU 2567/3-1.}%

\EventEditors{}
\EventNoEds{0}
\EventLongTitle{}
\EventShortTitle{}
\EventAcronym{}
\EventYear{}
\EventDate{}
\EventLocation{}
\EventLogo{}
\SeriesVolume{}
\ArticleNo{24}

\begin{document}

\pagenumbering{arabic}

\date{}

\maketitle

\begin{abstract}
A 2-packing set for an undirected graph $G=(V,E)$ is a subset $\mS \subset V$ such that any two
vertices $v_1,v_2 \in \mS$ have no common neighbors. Finding a 2-packing set of maximum cardinality is a
NP-hard problem. We develop a new approach to solve this problem on arbitrary graphs using its
close relation to the independent set problem. Thereby, our algorithm {\redp} uses new data reduction
rules specific to the 2-packing set problem as well as a graph transformation.
Our experiments show that we outperform the state-of-the-art for arbitrary graphs with respect to solution quality and also are able to compute solutions multiple orders of magnitude faster than previously possible. 
For example, we are able to solve \numprint{63}\% of the graphs in the tested data set to optimality in less than a second while the competitor for arbitrary graphs can only solve \numprint{5}\% of these graphs to optimality even with a 10 hour time limit. Moreover, our approach can solve a wide range of large instances that have \hbox{previously been unsolved.}
\end{abstract}

\maketitle
\vfill
\pagebreak 

\setcounter{page}{1}%

\section{Introduction}\label{sec:intro}
For a given undirected graph $G=(V,E)$ a \textit{2-packing set} is defined as a subset $\mS \subseteq V$ of all vertices such that for each pair of 
distinct vertices $v_1 \neq v_2 \in \mS$ the shortest path between $v_1$ and $v_2$ has at least length three.
A \textit{maximum 2-packing set (M2S)} is a 2-packing set with highest cardinality.
A generalization of the maximum 2-packing set problem is the maximum $k$-packing set problem, 
where the shortest path length is bounded by~$k$.
For $k=1$ this results in the \textit{maximum independent set (MIS)} problem.
An important application for the maximum 2-packing set problem is given for example in distributed algorithms. 
In contrast to the independent set problem, which given a solution vertex only conflicts the direct neighborhood, 
the 2-packing set provides information about a larger area around the vertex. This is important
for self-stabilizing algorithms 
\cite{gairing2004self,gairing2004distance,
        manne2006memory,shi2012self,turau2012efficient,
        trejo2012self,trejo2017self}.
In particular, computing large 2-packing sets can be used as a subroutine to ensure mutual exclusion of vertices
with overlapping neighborhoods. An example is finding a minimal \{$k$\}-dominating function~\cite{gairing2004self} which has various applications itself 
as presented by Gairing~\etal\cite{gairing2003exclusion}.
Bacciu~\etal\cite{bacciu2022generalizing} use large $k$-packing sets to develop a Downsampling approach for graph data. 
This is particularly useful for deep neural networks.
Further, Soto~\etal\cite{soto2018token} show that the knowledge of the size of a maximum 2-packing set in special graphs
can be used for error correcting codes and Hale~\etal\cite{hale1980frequence} indicate that large 2-packing sets
can be used to model interference issues for frequency assignment. This can be done by looking at the frequency-constrained co-channel assignment problem.
Here, the vertex set consists of locations of radio transmitters and two vertices share an edge if the frequencies are mutually perceptible. Now we want to assign a channel to as many radio transmitters as possible to conserve spectrum and to avoid interference issues vertices assigned the same channel must have a certain distance. For the distance two, this can be solved via finding a maximum 2-packing set in the corresponding graph. 
The maximum 2-packing set problem can be solved quite fast to optimality for small instances, however,
since it is a NP-hard problem, the running time of exact algorithm grows exponentially with 
the size of the graph. 
One powerful technique for tackling \NP{}-hard graph problems is to use \emph{data reduction rules}, which remove or contract local (hyper)graph structures, to reduce the input instance to an equivalent, smaller instance. Originally developed as a tool for parameterized algorithms~\cite{cygan2015parameterized}, data reduction rules have been effective in practice for computing an (unweighted) maximum independent set~\cite{chang2017computing,lamm2017finding,strash2016power} / minimum vertex cover~\cite{akiba-tcs-2016}, maximum clique~\cite{chang2020,verma2015}, and maximum $k$-plex~\cite{conte2021meta,jiang2021new}, as well as solving graph coloring~\cite{lin2017reduction,verma2015} and clique cover problems~\cite{gramm2009data,strash2022effective}, among others~\cite{Abu-Khzam2022}. To the best of our knowledge, no data reduction rules have yet been investigated for the 2-packing set problem.
The maximum 2-packing set problem can be solved with a graph transformation by expanding the 2 neighborhood and applying independent set solvers on the square graph. In recent years there was a lot of engineering yielding highly scalable solvers for maximum (weighted) independent set. However, graphs can get very dense if we directly compute the graph transformation which prohibits scalability. Thus problem specific data reductions are applied on the original problem instance first, and then we compute the expanded graph on which we solve the maximum independent set problem.

\textit{Our Results.}
Next to the new data reduction rules for the maximum 2-packing set problem, we also contribute a novel exact algorithm {\epack} as well as a heuristic algorithm {\hpack}. Both use
these reductions to solve the maximum 2-packing set problem for arbitrary graphs on large scale.
These algorithms work in three phases. First the data reductions are applied to the given graph
resulting in a reduced instance. Afterwards, the resulting graph is transformed,
such that a solution on the transformed graph for the MIS problem corresponds to a solution of
the maximum 2-packing set problem for the original graph. The third phase of the algorithms consists of solving the MIS problem on the transformed graph. In this phase our two variants differ. For {\epack} we use an exact solver, while we use a heuristic algorithm for {\hpack}.
Our experiments indicate that our algorithms outperform current state-of-the-art algorithms for arbitrary graphs both in terms of solution quality and running time. For instance, we can compute optimal solutions for 63\% of our graphs in under a second, whereas the competing method for arbitrary graphs achieves this only for 5\% of the graphs even with a 10-hour time frame. %
Lastly, our method solves many large instances that remained unsolved~before.

\section{Preliminaries}\label{sec:prelim}
Let ${G=(V, E)}$ be an undirected graph with ${n=|V|}$ and ${m=|E|}$, where ${V=\{0,...,n-1\}}$.
We extend this graph definition with further information about the 2-neighbors. Therefore, we define a graph $\G=(V,E,\E)$ with the set of vertices $V$ and edges $E$ as before. 
With the set $\E$ we include additional edges connecting vertices that share a common neighbor, but are not directly connected.
We only use the graph $\G$ when we emphasize the use of the additional edges $\E$, but all concepts introduced for the graph $G$ can be easily expanded for $\G$.
The neighborhood $\n{v}$ of a vertex $v\in V$ is defined as $\n{v} = \{u \in V : (u,v) \in E\}$ and $\nex{v} = \n{v} \cup \{v\}$.
The notation is extended to a set of vertices $U\subseteq V$ with 
$\n{U} = \cup_{u \in U}\n{u}\setminus U$ and $\nex{U} = \n{U}\cup U$.
Analogue we define the 2-neighborhood as the neighborhood via edges in $\E$ of a vertex $v\in V$ as $\twon{v} = \{u \in V : (u,v) \in \E\}$ and $\twone{v} = \twon{v} \cup \nex{v}$.
The notation is extended to a set of vertices $U\subseteq V$ with 
$\twon{U} = \cup_{u \in U}\twon{u}\setminus \nex{U}$ and $\twone{U}= \twon{U}\cup \nex{U}$.
These notations are analogue extended for $k$-neighborhoods.
 We define the subset of edges in $\E$ connecting the vertices with a common neighbor $v$ by
$\E(v)\coloneqq\{(x,y):x,y\in \twon{v} \ and \  (x,y)\notin E\}$.
The degree of a vertex $v$ is the size of its neighborhood $\deg(v) = |\n{v}|$. The 2-degree of a vertex is defined by the size of its 2-neighborhood $\degtwo(v) = |\twon{v}|$.
The \textit{square graph} $G^2=(V,E^2)$ of a graph $G=(V,E)$ is defined as a graph with the same vertex set and an extended edge set. Here, an edge $(u,v)$ is added if the distance between $u$ and $v$ is at most 2.
For ${0<k\in \mathbb{N}}$ a \textit{$k$-packing set} is defined as a subset
${\mS \subseteq V}$ of all vertices such that between each pair of vertices in $\mS$ the shortest path contains at least $k+1$ edges.
For $k=1$ we refer to the set as the \textit{independent set}, 
where all vertices in $\mS$ have to are non-adjacent.
The main focus in this work lies on the case $k=2$, 
where all vertices in $\mS$ must not have a common neighbor. 
A \textit{maximal} 2-packing set is a 2-packing set $\mS \subseteq V$ that can not be extended by any further
vertex $v \in V$ without violating the 2-packing set conditions.
The \textit{maximum} 2-packing set problem is that of finding a 2-packing set with maximum cardinality.
Analogue to the independence number $\alpha(G)$ for the maximum independent set, we define $\al(G)$ as the size of the solution to the maximum 2-packing set problem for a given graph $G$.

In a path a 2-edge contributes 2 to its length.
A 2-clique is as set of vertices in $\G$ such that they are pairwise connected by a path of at most length 2. A vertex $v$ is called \textit{2-isolated} if the vertices that are in its 2-neighborhood $\twone{v}$ form a 2-clique.

\section{Related Work} \label{sec:related_work}

\subsection{2-Packing Set Algorithms}
There are not many contributions to sequential algorithms for the 2-packing set problem
on arbitrary graphs. Trejo-Sánchez~\etal\cite{trejo2020genetic} are the first and 
as far as we know only authors to have proposed a sequential algorithm for connected arbitrary graphs. 
They developed a genetic algorithm for the M2S problem in which they use 
local improvements in each round of their algorithm and a penalty function.
Ding~\etal\cite{ding2014self} propose a self-stabilizing algorithm with safe convergence in an arbitrary graph. The algorithm consists basically of two 
operations, namely entering the solution candidate and exiting the solution candidate for each vertex in the graph. If a vertex enters the solution its neighbors 
get locked so they can not enter the solution and cause a conflict. The decision to enter or exit the solution is based on the simple criteria to check whether a
vertex causes a conflict or not. 
In general, most of the contributions to the 2-packing set problem on arbitrary graphs are in the context 
of distributed algorithms~\cite{gairing2004self, manne2006memory, shi2012self}. 
Further, there are some contributions to distributed algorithms for the M2S problem for special graphs.
Flores-Lamas~\etal\cite{flores2020distributed} present a distributed algorithm that 
finds a maximal 2-packing set in an undirected non-geometric Halin graph in linear time. They use reduction rules by Eppstein~\cite{eppstein2016halin} to determine 
a partition of the vertex set on which they base a coloring scheme to determine a maximal 2-packing set. However, the reduction of the graph is only used 
temporarily to compute the vertex partition and the coloring phase of the algorithm works on the original graph. 
Fernández-Zepeda~\etal\cite{trejo2014distributed} present a distributed algorithm for 
a M2S in geometric outerplanar graphs. 
Mjelde~\cite{mjelde2004packing} presents a self-stabilizing algorithm for the maximum $k$-packing set problem on tree graphs. Further, Mjelde~\cite{mjelde2004packing}
present a sequential algorithm for the M2S problem on tree graphs using dynamic programming.
For special graphs there are also some non-distributed algorithms for the M2S problem~\cite{soto2018token,flores2018algorithm,trejo2023approximation}. 
Trejo-Sánchez~\etal\cite{trejo2023approximation} present an approximation algorithm called \textsc{Apx-2p + Imp2p} using graph decompositions and LP-solvers. The approximation ratio is related
to how the algorithm decomposes the input graph into smaller subgraphs which is inspired by Baker~\cite{baker1994approximation}. They also mention the possibility to solve the maximum 2-packing set problem by using a graph transformation to the square graph and apply independent set solver. However, this is not used for their algorithm. This equivalence was stated by Halldórsson~\etal\cite{DBLP:conf/icalp/Telle98}.

\subsection{Independent Set Algorithms}
Since our algorithms use independent set algorithms we also summarize closely related work on this topic.  
As exact approaches to this problem mainly use data reduction techniques we summarize heuristic approaches for the unweighted as well as for the
weighted case because a lot of research happens in this area.
A widely used technique are local search algorithms. The main idea behind it is to start from an initial solution 
and then to improve this solution by insertions or swaps. The algorithm of Andrade~\etal\cite{andrade2012localsearch} proved to be a successful algorithm in this area.
The basic idea is to perform $(1,2)$-swaps, \ie to delete one vertex from the solution and add two new vertices. Hence, the solution size increases by one. 
Other local search approaches proved to be successful too~\cite{nogueira2018hybrid, cai2018cover, cai2017goodlocal, cai2015goodlocal2}. 
Further, there are also contributions that combine local search algorithms with data reduction~\cite{li2020redu} and graph neural networks~\cite{langedal2022neural}. There are also reduction based heuristics for the MIS problem. Lamm~\etal\cite{lamm2017scale} present a branch-and-reduce approach 
combined with an evolutionary algorithm. Further, 
Gu~\etal\cite{gu2021redu} use data reduction and a tie-breaking policy to apply data reductions repeatedly until they reach an empty graph.

\subsection{Data Reduction Rules}

Our algorithm uses data reductions in problem specific context as well as for solving the MIS problem on the transformed graph.
Therefore, we summarize some related work on this topic with respect to the MIS problem as well. In recent years the branch-and-reduce paradigm has been shown
to be an effective method to solve the maximum independent set problem to optimality as well as its complementary problem the minimum vertex~cover~\cite{akiba-tcs-2016}. By this paradigm we mean branching algorithms that use reduction rules to decrease the input size. Akiba and Iwata
show that this approach yields good results in comparison to other exact approaches for minimum vertex cover problem and the 
maximum independent set problem~\cite{akiba-tcs-2016}.
Further, this approach is successfully used for the maximum weighted independent set problem. Lamm~\etal\cite{lamm2019exactly} use this approach for an exact 
algorithm. Further, data reductions are also used for heuristics. Gro{\ss}mann~\etal\cite{DBLP:conf/gecco/GrossmannL0S23} use reduction rules in combination with 
an evolutionary approach for solving the maximum weighted independent set problem on huge sparse networks. Gao~\etal\cite{gao2017localredu} 
use inexact reduction rules by performing multiple rounds of a local search algorithm to determine vertices that are likely to be part of a solution to the MIS problem.

\section{The Algorithm {\redp}} \label{sec:algorithm}
\begin{wrapfigure}{T}{0.5\textwidth}
\vspace*{-0.85cm}
\begin{minipage}{0.5\textwidth}
\begin{algorithm}[H]
	\caption{High-level view of {\redp}. Structure of both \textsc{B\&R} and \textsc{Heuristic}.}\label{algo:red2pack}
	\begin{algorithmic}
		\STATE \textbf{input} graph ${\G=(V,E,\E)}$
		\STATE \textbf{output} M2S
		\STATE \textbf{procedure} {\redp}($\G$) 
		\STATE \quad $K \leftarrow$\textsc{reduce($\G$)} 
		\STATE \quad $K^2 \leftarrow$\textsc{transform($K$)} 
		\STATE \quad \textsc{MISsolve($K^2$)} 
	\end{algorithmic}
	\vspace*{-.1cm}
\end{algorithm}
\end{minipage}
\vspace*{-.3cm}
\end{wrapfigure}
We now give an overview of the components of our algorithm.
The idea of our approach is to build a square graph on which a maximum independent set is equivalent to a maximum 2-packing set on the original graph (see Theorem~\ref{thm: trafo equivalence}). 
On this transformed graph, we apply well studied maximum independent set solvers to find optimal solutions. During the transformation we increase the number of edges in the graph. Since this results in more dense graphs, this approach becomes quite slow and necessitates a substantial amount of memory.  
To alleviate this issue, we add a preprocessing step. There, we apply new problem specific data reductions exhaustively to the graph to obtain a kernel $\mathcal{K}$. We then apply the transformation on $\mathcal{K}$, resulting in a significantly smaller square graph. On this a maximum independent set solver is applied to obtain a (optimum) solution. In the end, the solution is transformed to a (optimum) solution to the input instance. Overall, this results in the algorithm {\redp} (see Algorithm~\ref{algo:red2pack}). In the following we introduce our new data reductions, then present the used graph transformation and finally give details about the maximum independent set solvers used.

\subsection{Data Reduction Rules} \label{sec:reductions}
To reduce the problem size especially for large instances exact data reductions are very useful tools.
In general our reductions allow us to identify vertices as (1) part of a solution to the maximum 2-packing set problem (included) and (2) as non-solution vertices (excluded). The reductions are applied exhaustively in a predefined order.
If a reduction is successful, we start again by applying the first reduction. 

We denote the reduced graph by $\G'$
and $\al(\G)$ is the solution size to the maximum 2-packing set problem for the given graph $\G$. For the \hbox{2-packing} set problem we first introduce the two \textit{core} data reductions in the following.

\begin{reduction}[Domination]\label{red:domination}
	Let ${v, u \in V}$ be vertices such that $\twone{v} \subseteq \twone{u}$. Then, there is some
	M2S of $\G$ that does not contain $u$. In particular, the graph is reduced
	to $\G' = \G[V\setminus \{u\}]$ and $\al(\G) = \al(\G')$.
\end{reduction}

\begin{proof}
	Let ${v, u \in V}$ and $\twone{v} \subseteq \twone{u}$.
	Further assume $\mS$ is a M2S in $\G$ containing $u$.
	Since $\twone{v}\subseteq \twone{u}$, it holds for all vertices $x\in \twone{v}\setminus\{u\}$
	that $x \notin \mS$. We define $\mS' = (\mS\setminus \{u\})\cup \{v\}$ and it holds
	$|\mS|=|\mS'|$. Further, $\mS'$ is still a valid 2-packing set since there is no vertex
	in $\twone{v}\setminus\{v\}$ that is also element of $\mS'$. By~construction $\mS'$  has the same size and therefore is a equivalent solution to M2S not containing $u$.
\end{proof}

\begin{reduction}[Clique] \label{red:clique}
	Let $v \in V$ be 2-isolated in $\G$. Then, $v$ is in some M2S of $\G$.
	Therefore, we can include $v$ and exclude $\twone{v}$.
	Resulting in $\G' = \G[V\setminus\twone{v}]$ and $\al(\G) = \al(\G')+1$.
\end{reduction}

\begin{proof}
	Let $\mS\subseteq V$ be a M2S in $\G$. Since $v$ is 2-isolated at least one vertex $w \in\twone{v}$ is
	contained in $\mS$, otherwise $\mS$ is not maximal. Assume $w\neq v$ and $\twone{v}$ forms a 2-clique, it holds ${u \notin \mS}$ for all
	${u \in \twone{v}\setminus \{w\}}$. Additionally, since $v$ is 2-isolated ${\twone{v} \subseteq \twone{w}}$. Therefore, we can always build a new solution $\mS' = (\mS\setminus \{w\})\cup \{v\}$ of the same size containing~$v$. This way we can always construct an equivalent or better solution when including $v$ and therefore the vertex $v$ is in some M2S of $\G$. Reducing the graph by including $v$ results in $\al(G) = 1+ \al(\G[V\setminus \twone{v}])$.
\end{proof}

These two reductions require knowledge about the 2-neighborhood as a prerequisite. Computing the 2-neighborhood of a vertex $v$ can become quite large and verifying these conditions is computationally expensive. Hence, we have sought out different special cases where it suffices to consider only the direct neighborhood and impose a constraint on the degree of the 2-neighborhood, which we maintain. Adding the following special cases to the \textit{core} data reductions yields our \textit{elaborated} data reductions. The following lemma helps us to show that these special cases are instances of the more general case.

\begin{lemma}\label{lem:deg}
		Let $u,v \in V$ be neighbors in $\G$ with $\nex{v}\subseteq\nex{u}$ such that $\degtwo(v) +\deg(v) \leq \deg(u)$. Then, $\twon{v} = \nex{u}\setminus\nex{v}$, \ie all 2-neighbors of the vertex $v$, are also neighbors of \hbox{the vertex $u$.}
\end{lemma}
\begin{proof}
		Let $u,v \in V$ be neighbors with $\nex{v}\subseteq\nex{u}$ such that $\degtwo(v)+\deg(v) \leq \deg(u)$. By definition of the 2-neighborhood we know that $\nex{u}\setminus\nex{v}\subseteq \twon{v}$. Therefore, we know that $\deg(u) + 1 -(\deg(v) + 1) \leq \degtwo(v)$ which is equivalent to $\deg(u) \leq \degtwo(v) +\deg(v)$. Because of the assumption $\degtwo(v) +\deg(v) \leq \deg(u)$, there has to hold an equality and therefore the sets have to be the same $\twon{v} = \nex{u}\setminus\nex{v}$.
\end{proof}

\begin{reduction}[Degree Zero Reduction] \label{red:deg0}
	Let $v \in V$ be a degree zero vertex. Furthermore, let $\degtwo(v) \leq 1$.
	Then, $v$ is in some M2S of~$\G$. Therefore, vertex $v$ can be included and the 2-neighborhood of $v$ excluded. Resulting in
	$\G'=\G[V\setminus\twone{v}]$ and \hbox{$\al(\G) = \al(\G')+1$.}
\end{reduction}

\begin{proof}
	Let ${v \in V}$ be a vertex with ${\deg(v)=0}$ and $\degtwo(v) \leq 1$. For the case of $\degtwo(v) =0$ there is no conflict and $v$ can be included in the solution. In the case of $\degtwo(v) =1$, let the 2-neighbor of $v$ be $u\in\twon{v}$. Assume $u$ is part of the solution $\mS$. Then, we can create a new solution by $\mS' = \mS\setminus \{u\} \cup \{v\}$ of same size. This way we always find a M2S including~$v$.
\end{proof}

\begin{figure}[t!]
    \centering
    \caption{Original graph on the left reduced to the graph on the right. Due to clarity reasons, we only display the necessary edges from $\E$ (dashed) regarding the 2-neighborhood information.}
    \subcaptionbox{ Domination Reduction: 
        $\twone{u}$ is colored with orange, $\twone{v}$ with blue.
        The vertex $u$ dominates $v$ and therefore $u$ can be excluded.
        \label{fig:dom}}[.48\textwidth]{
        \begin{tikzpicture}[scale=0.35]

\begin{scope}[every node/.style={circle, thick,draw=black}]
\node[draw=white,label={[color=blue,xshift=0.0cm,yshift=0cm]{$v$}}] (v0) at (3,7.5) {};
\node[draw=white,label={[color=orange,xshift=0.0cm,yshift=0cm]{$u$}}] (u0) at (7,7.5) {};
\node[] (v) at (3,7.5) {};
\node[] (u) at (7,7.5) {};
\node[] (v4) at (8.5,5.5) {};
\node[] (v5) at (4.5,5.5) {};
\node[] (v6) at (6.5,5.5) {};
\node[] (v7) at (1.5,5.5) {};
\node[] (v1) at (3,3.5) {};
\node[fill = orange] (v2) at (5,3.5) {};
\node[fill = orange] (v3) at (7,3.5) {};
\node[white] (G0) at (2,2) {};
\node[white] (G1) at (3,1.5) {};
\node[white] (G2) at (4,2) {};
\node[white] (G4) at (6,2) {};
\node[white] (G3) at (7,2) {};
\node[white] (G5) at (8,2) {};
\end{scope}

\begin{scope}[every node/.style={circle, thick,draw=black, fill=black}]
\node[draw=white,fill=white,label={[xshift=0.0cm,yshift=0cm]{$v$}}] (vk0) at (13,7.5) {};
\node[] (vk) at (13,7.5) {};
\node[draw=white,color=white,label={[color=redgray,xshift=0.0cm,yshift=0cm]{$u$}}] (uk0) at (17,7.5) {};
\node[color = redgray] (uk) at (17,7.5) {};
\node[] (vk1) at (13,3.5) {};
\node[] (vk2) at (15,3.5) {};
\node[] (vk3) at (17,3.5) {};
\node[] (vk4) at (18.5,5.5) {};
\node[] (vk5) at (14.5,5.5) {};
\node[] (vk6) at (16.5,5.5) {};
\node[] (vk7) at (11.5,5.5) {};
\node[white] (GI0) at (12,2) {};
\node[white] (GI1) at (13,1.5) {};
\node[white] (GI2) at (14,2) {};
\node[white] (GI4) at (16,2) {};
\node[white] (GI3) at (17,2) {};
\node[white] (GI5) at (18,2) {};

\end{scope}

\begin{scope}[>={Stealth[black]},
every edge/.style={draw=black,thick, color=black}]
\path [-] (v4) edge node {} (v3);
\path [-] (v5) edge node {} (v1);
\path [-] (v7) edge node {} (v1);
\path [-] (v6) edge node {} (v2);
\path [-] (u) edge node {} (v5);
\path [-] (u) edge node {} (v6);
\path [-] (u) edge node {} (v4);
\path [-] (v) edge node {} (v7);
\path [-] (v) edge node {} (u);

\path [-] (G1) edge node {} (v1);
\path [-] (G2) edge node {} (v2);
\path [-] (G4) edge node {} (v2);
\path [-] (G3) edge node {} (v3);

\path [-] (vk4) edge node {} (vk3);
\path [-] (vk5) edge node {} (vk1);
\path [-] (vk7) edge node {} (vk1);
\path [-] (vk6) edge node {} (vk2);
\path [-] (vk) edge node {}  (vk7);

\draw[lightgray,-,thick] (uk) --(vk4);
\draw[lightgray,-,thick] (uk) --(vk5);
\draw[lightgray,-,thick] (uk) --(vk6);
\draw[lightgray,-,thick] (uk) --(vk7);
\draw[lightgray,-,thick] (uk) --(vk);
\draw[lightgray,-,thick] (vk) -- (uk);

\end{scope}

\begin{scope}[thick, gray, dashed]
\draw[] (v) --(v4);
\draw[] (v) --(v5);
\draw[] (v) --(v6);

\draw[] (v5) --(v6);
\draw[] (v4) --(v5);
\draw[] (v5) arc (220:320:2.5);

\draw[] (vk) --(vk4);
\draw[] (vk) --(vk5);
\draw[] (vk) --(vk6);

\draw[] (vk4) --(vk6);
\draw[] (vk4) --(vk5);
\draw[] (vk5) arc (220:320:2.5);
\end{scope}

\begin{scope}[every node/.style={circle, thick,draw=black, fill=black}]
\node[] (vkk4) at (18.5,5.5) {};
\node[] (vkv5) at (14.5,5.5) {};
\node[] (vkk6) at (16.5,5.5) {};
\end{scope}
\begin{scope}
\path [-] (GI1) edge node {} (vk1);
\path [-] (GI2) edge node {} (vk2);
\path [-] (GI4) edge node {} (vk2);
\path [-] (GI3) edge node {} (vk3);

\fill[white] (1.5,2.5) -- (8.5,2.5) -- (8.5,1.5) -- (1.5,1.5) -- (1.5, 2.5);
\draw[black, thick] (1.5,2.5) -- (8.5,2.5);
\draw[black, thick] (8.5,1) -- (1.5,1);
\draw[black, thick] (1.5,2.5) arc (90:270:0.75);
\draw[black, thick] (8.5,1) arc (270:450:0.75);

\fill[white] (11.5,2.5) -- (18.5,2.5) -- (18.5,1.5) -- (11.5,1.5) -- (11.5, 2.5);
\draw[white] (G1) -- (G3)  node[midway, black] {$G[V\setminus \twone{\{v,u\}}]$} ;
\draw[white] (GI1) -- (GI3)  node[midway, black] {$G[V\setminus \twone{\{v,u\}}]$} ;

\draw[black, thick] (11.5,2.5) -- (18.5,2.5);
\draw[black, thick] (18.5,1) -- (11.5,1);
\draw[black, thick] (11.5,2.5) arc (90:270:0.75);
\draw[black, thick] (18.5,1) arc (270:450:0.75);

\draw[-Stealth, thick] (9.5,5.5) -- (10.5,5.5); 

\fill[orange] (3,7) arc   (270:450:0.5);
\fill[orange] (7,7) arc   (270:450:0.5);
\fill[orange] (3,3) arc   (270:450:0.5);
\fill[orange] (6.5,5) arc   (270:450:0.5);
\fill[orange] (4.5,5) arc   (270:450:0.5);
\fill[orange] (1.5,5) arc (270:450:0.5);
\fill[orange] (8.5,5) arc (270:450:0.5);

\fill[blue] (3,8)   arc (90:270:0.5);
\fill[blue] (7,8)   arc (90:270:0.5);
\fill[blue] (3,4)   arc (90:270:0.5);
\fill[blue] (6.5,6)   arc (90:270:0.5);
\fill[blue] (4.5,6)   arc (90:270:0.5);
\fill[blue] (1.5,6) arc (90:270:0.5);
\fill[blue] (8.5,6) arc (90:270:0.5);

\begin{scope}[every node/.style={circle, thick,draw=black}]
	\node[] (v) at (3,7.5) {};
	\node[] (u) at (7,7.5) {};
	\node[] (v4) at (8.5,5.5) {};
	\node[] (v5) at (4.5,5.5) {};
	\node[] (v6) at (6.5,5.5) {};
	\node[] (v7) at (1.5,5.5) {};
	\node[] (v1) at (3,3.5) {};
\end{scope}

\end{scope}
\end{tikzpicture}}\hfill
    \subcaptionbox{Twin Reduction: 
        The neighbors $u,w \in V$ of $v$ are twins. We can include $v$ into the solution and exclude its 2-neighborhood $\twone{v}$.
        \label{fig:twin}}[.48\textwidth]{
        \begin{tikzpicture}[scale=0.35]
\draw[-Stealth, thick] (9.5,3.5) -- (10.5,3.5);

\begin{scope}[every node/.style={circle,scale=1.,thick,draw=black, fill=black}]
\node[white,label={[xshift=0.0cm,yshift=0cm]$v$}] (v0) at (5,5.5) {};
\node[white,label={[xshift=0.0cm,yshift=0cm]$w$}] (v40) at (8.5,5.5) {};
\node[white,label={[xshift=0.0cm,yshift=0cm]$u$}] (v50) at (1.5,5.5) {};
\node (v) at (5,5.5) {};
\node (v4) at (8.2,5.5) {};
\node (v5) at (1.8,5.5) {};
\node (v1) at (3,3.5) {};
\node (v2) at (5,3.5) {};
\node (v3) at (7,3.5) {};
\node (g1) at (3,2.) {};
\node (g2) at (4,1.2) {};
\node (g3) at (6,1.2) {};
\node (g4) at (7,2.) {};
\node[white] (G1) at (3,-0.5) {};
\node[white] (G2) at (5,0) {};
\node[white] (G3) at (5,0) {};
\node[white] (G4) at (7,0) {};

\node[scale=1.5,white,label={[greengray,xshift=0.0cm,yshift=0cm]$v$}] (x) at (15,5.5) {};
\node[scale=1.5,white,label={[redgray,xshift=0.0cm,yshift=0cm]$w$}] (x4) at (18.5,5.5) {};
\node[scale=1.5,white,label={[redgray,xshift=0.0cm,yshift=0cm]$u$}] (x5) at (11.5,5.5) {};
\node[greengray] (x) at (15,5.5) {};
\node[redgray] (x4) at (18.2,5.5) {};
\node[redgray] (x5) at (11.8,5.5) {};
\node[redgray, draw=redgray] (x1) at (13,3.5) {};
\node[redgray, draw=redgray] (x2) at (15,3.5) {};
\node[redgray, draw=redgray] (x3) at (17,3.5) {};
\node (X1) at (13,2.) {};
\node (X2) at (14,1.2) {};
\node (X3) at (16,1.2) {};
\node (X4) at (17,2.) {};
\node[white] (XX1) at (13,-0.5) {};
\node[white] (XX2) at (15,0) {};
\node[white] (XX3) at (15,0) {};
\node[white] (XX4) at (17,0) {};

\end{scope}

\begin{scope}[>={Stealth[black]}, thick,
every edge/.style={draw=black,thick, color=black}]
\path [-] (v1) edge node {} (g1);
\path [-] (v2) edge node {} (g2);
\path [-] (v2) edge node {} (g3);
\path [-] (v3) edge node {} (g4);
\path [-] (v4) edge node {} (v1);
\path [-] (v4) edge node {} (v2);
\path [-] (v4) edge node {} (v3);
\path [-] (v5) edge node {} (v1);
\path [-] (v5) edge node {} (v2);
\path [-] (v5) edge node {} (v3);
\path [-] (v) edge node {} (v5);
\path [-] (v) edge node {} (v4);

\path [-] (G1) edge node {} (g1);
\path [-] (G3) edge node {} (g2);
\path [-] (G2) edge node {} (g3);
\path [-] (G4) edge node {} (g4);

\draw[lightgray,-] (x1) -- (X1);
\draw[lightgray,-] (x2) -- (X2);
\draw[lightgray,-] (x2) -- (X3);
\draw[lightgray,-] (x3) -- (X4);
\draw[lightgray,-] (x4)-- (x1);
\draw[lightgray,-] (x4)-- (x2);
\draw[lightgray,-] (x4)-- (x3);
\draw[lightgray,-] (x5)-- (x1);
\draw[lightgray,-] (x5)-- (x2);
\draw[lightgray,-] (x5)-- (x3);
\draw[lightgray,-] (x) --(x5);
\draw[lightgray,-] (x) --(x4);

\path [-] (XX1) edge node {} (X1);
\path [-] (XX3) edge node {} (X2);
\path [-] (XX2) edge node {} (X3);
\path [-] (XX4) edge node {} (X4);

\draw[thick,gray,dashed] (g3)--(g2);

\draw[thick,gray,dashed] (X2)--(X3);

\fill[white] (12.5,0.5) -- (17.5,0.5) -- (17.5,-.5) -- (12.5,-.5) -- (12.5, 0.5);
\fill[white] (2.5,0.5) -- (7.5,0.5) -- (7.5,-.5) -- (2.5,-.5) -- (2.5, 0.5);

\draw[black, thick] (2,0.5) -- (8,0.5);
\draw[black, thick] (8,-1) -- (2,-1);
\draw[black, thick] (12,0.5) -- (18,0.5);
\draw[black, thick] (18,-1) -- (12,-1);
\draw[black, thick] (2,0.5) arc (90:270:0.75);
\draw[black, thick] (8,-1) arc (270:450:0.75);
\draw[black, thick] (12,0.5) arc (90:270:0.75);
\draw[black, thick] (18,-1) arc (270:450:0.75);

\draw[white] (G1) -- (G4)  node[scale=1,midway, black] {$G[V\setminus \twone{v}]$} ;
\draw[white] (XX1) -- (XX4)  node[scale=1,midway, black] {$G[V\setminus \twone{v}]$} ;

\end{scope}
\end{tikzpicture}}
\end{figure}

\begin{reduction}[Degree Zero Triangle] \label{red:deg0triangle}
	Let $v \in V$ be a degree zero vertex. Furthermore, let $\degtwo(v) = 2$ with 2-neighbors $\twon{v}=\{u,w\}$ that are also adjacent via an edge or 2-edge $u\in \twone{w}$.
	Then, $v$ is in some M2S of $\G$. Therefore, vertex $v$ can be included and the 2-neighborhood of $v$ excluded. This results in
	$\G'=\G[V\setminus\twone{v}]$ and $\al(\G) = \al(\G')+1$.
\end{reduction}

\begin{proof}
	Let ${v \in V}$ be a vertex of ${\deg(v)=0}$ and $\{u,w\}=\twon{v}$ its 2-neighbors adjacent via an edge or 2-edge $u\in \twone{w}$. Vertices $u$ and $w$ dominate vertex $v$ and can therefore be excluded by Reduction~\ref{red:domination}. Now Reduction~\ref{red:deg0} is applicable and vertex~$v$ can be included into the solution.
\end{proof}

\begin{reduction}[Degree One] \label{red:deg1}
    Let $v \in V$ be a degree one vertex with ${\n{v} = \{u\}}$. Furthermore, let $\degtwo(v) \leq \deg(u)-1$.
    Then, $v$ is in some M2S of $\G$. Therefore, vertex $v$ can be included and the 2-neighborhood of $v$ excluded. This results in
    $\G'=\G[V\setminus\twone{v}]$ and $\al(\G) = \al(\G')+1$.
\end{reduction}

\begin{proof}
    Let ${v \in V}$ be a vertex with ${\deg(v)=1}$ and ${\n{v} = \{u\}}$ with $\degtwo(v) \leq \deg(u)-1$. We can apply Lemma~\ref{lem:deg} and it holds that $\twone{v} = \nex{u}$. 
    Therefore, this represents a special case of the 2-clique and Reduction~\ref{red:clique} 
    can be applied.
\end{proof}

\begin{reduction}[Degree Two V-Shape]\label{red:v-shape}
	Let ${v \in V}$ be a vertex of ${\deg(v)=2}$ with $N(v) = \{u,w\}$ and $\degtwo(v)=0$.
	Then, $v$ is in some M2S of $\G$.
	The graph is reduced to $\G' = \G[V\setminus\twone{v}]$ and $\al(\G) =\al(\G') + 1$.
\end{reduction}

\begin{proof}
	Let the above stated assumptions hold. Then, $\twone{v} \subseteq \twone{u}$ and $\twone{v} \subseteq \twone{w}$ the vertices $w$ and $u$ can be excluded by Reduction~\ref{red:domination}. Since $|\twon{v}|=0$, vertex~$v$ can \hbox{safely be included.}
\end{proof}

\begin{reduction}[Degree Two Triangle]\label{red:triangle}
	Let ${v \in V}$ be a vertex of ${\deg(v)=2}$ with $N(v) = \{u,w\}$ and $2=\deg(u)=\deg(w)$. Furthermore, let $\degtwo(v)=0$.
	Then $v$ is in some M2S of $\G$.
	The graph is reduced to $\G' = \G[V\setminus\twone{v}]$ and $\al(\G) =\al(\G') + 1$.
\end{reduction}

\begin{proof}
	Let the vertices $v,u,w \in V$ all have degree two and $\n{v} = \{u,w\}$ and $\degtwo(v)=0$. In this case, the vertices $u$, $v$ and $w$ form a triangle. Since $\twone{v} \subseteq \twone{u}$ and $\twone{v} \subseteq \twone{w}$ the vertices $w$ and $u$ can be excluded by Reduction~\ref{red:domination}. Since $|\twon{v}|=0$, vertex~$v$ can \hbox{safely be included.}
\end{proof}

\begin{reduction}[Degree Two 4-Cycle]\label{red:cycle}
	Let ${u,v,w,x \in V}$ be vertices with $\deg(v)=\deg(u)=\deg(w)=2$, $\n{v} = \{u,w\}$ and $\twon{v} =\{x\}$. Furthermore, let $x\in\n{u}$ and $x\in\n{w}$. Then, the vertices build a 4-cycle and $v$ is in some M2S of $\G$.
	The graph is reduced to $\G' = \G[V\setminus\twone{v}]$ and $\al(\G) =\al(\G') + 1$.
\end{reduction}

\begin{proof}
	Assuming the above stated assumptions such that the vertices $u$, $v$, $w$ and the one 2-neighbor $x\in\twon{v}$ form a 4-cycle. It holds that $\twone{v}=\{u,v,w,x\}\subseteq\twone{u}$, therefore, we can exclude the vertex $u$ by Reduction~\ref{red:domination}. Analogue, we can exclude vertex $w$.  
	Assume $x$ is part of a M2S $\mS$. Then, we can create a new solution $\mS' = \mS\setminus \{x\} \cup \{v\}$ of same size. This way we always find a M2S including~$v$.
\end{proof}

\begin{reduction}[Fast Domination]\label{red:fastDomination}
	Let ${v, u \in V}$ be vertices such that $\nex{v} \subseteq \nex{u}$ and $\degtwo(v) + \deg(v) \leq \deg(u)$. Then, there is some M2S of $\G$ without $u$. The graph is reduced
	to $\G' = \G[V\setminus \{u\}]$ and $\al(\G) = \al(\G')$.
\end{reduction}

\begin{proof}	
	Let ${v, u \in V}$ and $\nex{v} \subseteq \nex{u}$ with $\degtwo(v) + \deg(v) \leq \deg(u)$. Using Lemma~\ref{lem:deg} we know $\twon{v}=\nex{u}\setminus\nex{v}$. It holds that $\twone{v} \subseteq \twone{u}$ and \hbox{Reduction~\ref{red:domination} is applicable.}
\end{proof}

\begin{reduction}[Twin] \label{red:twin}
	Let ${v \in V}$ be a vertex of degree ${\deg(v)=2}$ and 
	${u, w \in V}$ be its neighbors with $N(u) = N(w)$. Furthermore, let $\degtwo(v) \leq \deg(u)-1$.
	Then, $u$ and $w$ are twins and $v$ is in some M2S of $\G$. We reduce $\G'=\G[V\setminus\twone{v}]$
	and $\al(\G) = \al(\G')+1$.
\end{reduction}
\begin{proof}
	Let ${v \in V}$ be a vertex with ${\deg(v)=2}$ and 
	${u, w \in V}$ be its neighbors with $N(u) = N(w)$. The additional assumption $\degtwo(v) \leq \deg(u)-1 = \deg(w)-1$ ensures that all 2-neighbors of $v$ are connected to the vertices $u$ and $w$. Analogue to Lemma~\ref{lem:deg} but for the neighborhood of $u$ and $v$ not a subset.
	This way $\G[\twone{v}]$ forms a 2-clique, and therefore, this reduction is a special \hbox{case of
	Reduction~\ref{red:clique}.}
\end{proof}

\subsubsection{Miscellaneous} \label{sec:algoReductions}

We now motivate the importance of the edge set $\E$. 
When a vertex $v$ is included by a reduction, its 2-neighborhood has to be excluded to obtain a valid maximum 2-packing set.
In Figure~\ref{fig:toBeExcludedWrong} we perform Reduction~\ref{red:deg1} on the original graph $G$ without the separate 2-neighbor edges $\E$. Here, we do not get a valid solution since the information, that $n_1,\ n_2$ and $n_3$ are in the same 2-neighborhood is lost.
When we apply the reduction on $\G$ we do not lose connecting edges in $\E(w)$ of an excluded vertex $w$.
This is illustrated in Figure~\ref{fig:toBeExcludedRight}.

\begin{figure*}
    \centering
    \caption{
    	Demonstration of the vital role of retaining 2-neighborhood information for non-reduced vertices in maintaining a valid solution to the M2S problem. Reduced vertices and edges are light gray shaded. Green vertices are included, red vertices are excluded from the solution.    
    }\label{fig:toBeExcluded}
    \subcaptionbox{Degree One Reduction without $\E$ loosing the 2-neighborhood information yielding an invalid solution.
        \label{fig:toBeExcludedWrong}}[\textwidth]{\begin{tikzpicture}[scale=0.45]
\draw[white] (7.5,5) -- (10.5,5) node[midway,black] {reduce} ;
\draw[white] (19,5) -- (22,5) node[midway,black] {solve} ;
\draw[-Stealth, thick] (7.5,4.5) -- (10.5,4.5);
\draw[-Stealth, thick] (19,4.5) -- (22,4.5);

\newcommand{\drawmygraph}[6]{
    \begin{scope}[shift={#2}, every node/.style={circle,scale=.8,thick,draw=black, fill=black}, every edge/.style={ultra thick}, scale=.8]
	\node[color=#4, label={[scale=1, color=#4,xshift=-0.1cm,yshift=-0.15cm]$v$}] (inv#1) at (2.,4.5) {};
	\node[color=#5, label={[scale=1, color=#5,xshift=-0.1cm,yshift=-0.15cm]$u$}] (exv#1) at (3.5,4.5) {};
	\node[color=#5, label={[scale=1, color=#5,xshift=-0.1cm,yshift=-0.15cm]$w$}] (vv#1) at (5,4.5) {};
	\node[label={[scale=1, color=#6,yshift=-.5cm, xshift=.5cm]$n_1$},#6] (n1v#1) at (6.5,5.6) {};
	\node[label={[scale=1, color=#6,yshift=-.6cm, xshift=.5cm]$n_2$},#6] (n2v#1) at (7.5,4.5) {};
	\node[label={[scale=1, color=#6,yshift=-.6cm, xshift=.5cm]$n_3$},#6] (n3v#1) at (6.5,3.4) {};
	\draw[-,#3, thick] (inv#1) -- (exv#1);
	\draw[-,#3, thick] (vv#1) -- (exv#1);
	\draw[-,#3, thick] (vv#1) -- (n1v#1);
	\draw[-,#3, thick] (vv#1) -- (n2v#1);
	\draw[-,#3, thick] (vv#1) -- (n3v#1);
	\end{scope}
}
\drawmygraph{1}{(-1,1)}{black}{black}{black}{black};
\drawmygraph{2}{(10.5,1)}{lightgray}{greengray}{redgray}{black};
\drawmygraph{3}{(22,1)}{lightgray}{greengray}{redgray}{green};

\end{tikzpicture}}
    \subcaptionbox{Degree One Reduction with necessary 2-neighborhood information via $\E$ resulting in \hbox{a valid M2S.}
        \label{fig:toBeExcludedRight}}[\textwidth]{
        \begin{tikzpicture}[scale=0.45]
\pgfdeclarelayer{background}
\pgfdeclarelayer{main}
\pgfsetlayers{background,main}
	
\draw[white] (7.5,5) -- (10.5,5) node[midway,black] {reduce} ;
\draw[white] (19,5) -- (22,5) node[midway,black] {solve} ;
\draw[-Stealth, thick] (7.5,4.5) -- (10.5,4.5);
\draw[-Stealth, thick] (19,4.5) -- (22,4.5);

\newcommand{\drawmygraph}[7]{
	\begin{scope}[shift={#2}, every node/.style={circle,scale=.8,thick,draw=black, fill=black}, scale=.8]
		\node[color=#4, label={[scale=1, color=#4,xshift=-0.1cm,yshift=-0.15cm]$v$}] (inv#1) at (2.,4.5) {};
		\node[color=#5, label={[scale=1, color=#5,xshift=-0.1cm,yshift=-0.15cm]$u$}] (exv#1) at (3.5,4.5) {};
		\node[color=#5, label={[scale=1, color=#5,xshift=-0.1cm,yshift=-0.15cm]$w$}] (vv#1) at (5,4.5) {};
		\node[label={[scale=1, color=#6,yshift=-.5cm, xshift=.5cm]$n_1$},#6] (n1v#1) at (6.5,5.6) {};
		\node[label={[scale=1, color=black,yshift=-.6cm, xshift=.5cm]$n_2$},black] (n2v#1) at (7.5,4.5) {};
		\node[label={[scale=1, color=black,yshift=-.6cm, xshift=.5cm]$n_3$},black] (n3v#1) at (6.5,3.4) {};
		\begin{pgfonlayer}{background}
		\draw[-,#3,  thick] (inv#1) -- (exv#1);
		\draw[-,#3,  thick] (vv#1) -- (exv#1);
		\draw[-,#3,  thick] (vv#1) -- (n1v#1);
		\draw[-,#3,  thick] (vv#1) -- (n2v#1);
		\draw[-,#3,  thick] (vv#1) -- (n3v#1);
		\draw[gray, dashed,  thick] (n1v#1) -- (n2v#1);
		\draw[gray, dashed,  thick] (n1v#1) -- (n3v#1);
		\draw[gray, dashed,  thick] (n2v#1) -- (n3v#1);
		\draw[#7, dashed,  thick] (n1v#1) -- (exv#1);
		\draw[#7, dashed,  thick] (n3v#1) -- (exv#1);
		\draw[#7, dashed,  thick] (inv#1) arc [start angle=180, end angle= 360, radius=1.5cm];
		\draw[#7, dashed,  thick] (exv#1) arc [start angle=180, end angle=0, radius=2cm];
		\end{pgfonlayer}
	\end{scope}
}
\drawmygraph{1}{(-1,1)}{black}{black}{black}{black}{gray};
\drawmygraph{2}{(10.5,1)}{lightgray}{greengray}{redgray}{black}{lightgray};
\drawmygraph{3}{(22,1)}{lightgray}{greengray}{redgray}{green}{lightgray};

\end{tikzpicture}}
\end{figure*}

\subsubsection{Data Structure} \label{sec:data structure}
Our reductions operate on a dynamic graph data structure based on adjacency arrays representing undirected edges by two directed edges.
Internally, we separate edges and 2-edges with two adjacency arrays.
A key observation of our reductions is that we often do not need to know the 2-neighborhood to exclude a vertex. This can be seen for example in the Degree One Reduction applied to a vertex $v$ with its neighbor $u$. Then, we do not need to consider $\twon{u}$.
In general, we only compute and then store the 2-neighborhood of a vertex $v$ on demand.

This also leads to less effort, compared to initially constructing the 2-neighborhood for the whole graph, when reducing the graph.
When a vertex $u$ is removed from the graph, we delete every edge and 2-edge pointing to it.
Removing all incoming edges can take $\mathcal{O}(\Delta^4)$ time where $\Delta$ is the highest degree in the graph. With the on-demand technique we can also reduce the amount of computations there, \ie
if the 2-neighborhood of a vertex $v$ is not computed in advance there are potentially less edges that are deleted in this step.

\subsection{Graph Transformation}\label{sec:trafo}
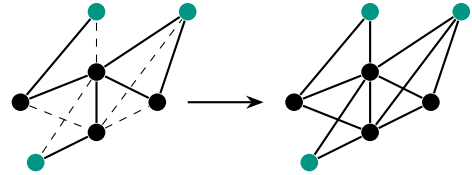
\begin{wrapfigure}{}{0.5\textwidth}
	\centering
		\vspace*{-.5cm}
	\caption{The graph ${\G=(V,E,\E)}$ on the left 
		with edges in $\E$ marked as dashed lines. This graph is transformed to ${G^2=(V,E^2)}$ on the right. The green vertices present the M2S (left) and \hbox{the MIS (right).}}\label{fig:graphTrafo}
	\begin{tikzpicture}[scale=.4]
\begin{scope}[every node/.style={circle,scale=0.7,thick,fill=black}]
\node[fill=green] (v1) at (3,1) {};
\node[fill=green] (v4) at (8,6) {};
\node[fill=green] (v6) at (5,6) {};
\node (v5) at (5,4) {};
\node (v7) at (2.5,3) {};
\node (v2) at (5,2) {};
\node (v3) at (7,3) {};

\node[fill=green] (g1) at (12,1) {};
\node[fill=green] (g4) at (17,6) {};
\node[fill=green] (g6) at (14,6) {};
\node (g5) at (14,4) {};
\node (g2) at (14,2) {};
\node (g3) at (16,3) {};
\node (g7) at (11.5,3) {};
\end{scope}

\begin{scope}[>={Stealth[black]},
every edge/.style={draw=black,thick, color=black}]
\path [-] (v2) edge node {} (v1);
\path [-] (v5) edge node {} (v2);
\path [-] (v5) edge node {} (v3);
\path [-] (v5) edge node {} (v4);
\path [-] (v4) edge node {} (v3);
\path [-] (v7) edge node {} (v5);
\path [-] (v7) edge node {} (v6);

\path [-] (g2) edge node {} (g1);
\path [-] (g5) edge node {} (g2);
\path [-] (g5) edge node {} (g3);
\path [-] (g5) edge node {} (g4);
\path [-] (g4) edge node {} (g3);
\path [-] (g7) edge node {} (g5);
\path [-] (g7) edge node {} (g6);

\path [-] (g5) edge node {} (g6);
\path [-] (g5) edge node {} (g1);
\path [-] (g2) edge node {} (g3);
\path [-] (g7) edge node {} (g2);
\path [-] (g2) edge node {} (g4);

\draw[-Stealth, thick] (8,3) -- (10.5,3);
\end{scope}

\begin{scope}[edge/.style={draw=black,thick, color=black}, dashed]
\path [dashed] (v5) edge node {} (v6);
\path [dashed] (v5) edge node {} (v1);
\path [dashed] (v2) edge node {} (v3);
\path [dashed] (v7) edge node {} (v2);
\path [dashed] (v2) edge node {} (v4);
\end{scope}

\end{tikzpicture}
			\vspace*{-.3cm}
\end{wrapfigure}
	After we applied all data reduction rules exhaustively, i.e.~there is no data reduction rule that can still be applied, we begin the transformation of the reduced graph to the square~graph. 
	
	For this process we start with a graph ${G=(V,E)}$ and add all edges in $\E$ as regular edges. The resulting set of edges is referred to as $E^2 = E\cup \E$, and the transformed square graph is denoted as ${G^2=(V,E^2)}$.
	This way the set of regular edges is extended to also contain edges connecting pairs of vertices $u,w \in V$ that have a common neighbor. Since $E \cap \E = \emptyset$ we do not add parallel edges.
	Figure~\ref{fig:graphTrafo} illustrates of the square graph. 

\begin{theorem}\label{thm: trafo equivalence}
	Let ${G=(V,E)}$ be a given graph and ${G^2=(V,E^2)}$ the transformed 
	graph.
	Then, an optimal solution to the maximum independent set problem on $G^2$ 
	is an optimal solution for the maximum 2-packing set problem on the original graph $G$ \cite{DBLP:conf/icalp/Telle98}.
\end{theorem}

\subsection{MIS Solver} \label{sec:misSolver}
We have chosen to use the solver by Lamm~\etal\cite{lamm-2019} since it is a state-of-the-art exact solver for the independent set problem. However, it is also possible to integrate any other exact solver for the maximum independent set problem. Note that we did not chose the branch and reduce solver for the unweighted problem \cite{redumis-2017}, since with it we are restricted to smaller graphs. While our focus is on optimal solution, we also combine our solver with \textsc{OnlineMIS}~\cite{DBLP:conf/wea/DahlumLS0SW16}, resulting in our heuristic {\hpack}.
For a briefly description of the main components of the solver by Lamm~\etal\cite{lamm-2019} and the solver by Dahlum~\etal\cite{DBLP:conf/wea/DahlumLS0SW16} we refer to our technical report~\cite{DBLP:journals/corr/abs-2308-15515}.

\section{Experimental Evaluation} \label{sec:experiments}

\textbf{Methodology.}
We implemented our algorithm using C++17. The code is compiled using g++ version 12.2 and full optimizations turned on (-O3).
We used a machine equipped with a AMD EPYC 7702P (64 cores) processor and 1 TB RAM running Ubuntu 20.04.1. %
We ran all of our experiments with four different random seeds and report geometric mean values unless mentioned otherwise. We set the time limit for all algorithms to 10h. If a solver exceeded a memory threshold of \numprint{100} GB during execution we stop the solver and mark this with m.o.. If the algorithm terminated due to the time limit, we mark it with t.o.~in the results. In both cases we report the best solution found \hbox{until this point.}
To compare different algorithms, we use performance profiles~\cite{dolan2002benchmarking}, which depict the relationship between the objective function size or running time of each algorithm and the corresponding values produced or consumed by the algorithms. Specifically, the y-axis represents the fraction of instances where the objective function is less than or equal to $\tau$ times the best objective function value, denoted as $\#\{\textnormal{objective} \leq \tau \ast \textnormal{best}\}/\#G$. Here, ``objective`` refers to the result obtained by an algorithm on an instance, and ``best" corresponds to the best result among all the algorithms shown in the plot. \#$G$ is the number of graphs in the data set.
When considering the running time, the y-axis displays the fraction of instances where the time taken by an algorithm is less than or equal to $\tau$ times the time taken by the fastest algorithm on that instance, denoted as $\#\{\textnormal{t}\leq \tau \ast \textnormal{ fastest}\}/\#\textnormal{instances}$. Here, ``t`` represents the time taken by an algorithm on an instance, and ``fastest" refers to the time taken by the fastest algorithm on that specific instance.
The parameter $\tau$, which is greater than or equal to 1, is plotted on the x-axis. Each algorithm's performance profile yields a non-decreasing, \hbox{piecewise constant function.} %

\begin{table}[b]
 \centering
 	\caption{Effect of graph transformation: Arithmetic mean percentage of number of vertices $\tilde{n} = n(K^2)/n(G)$ and edges $\tilde{m}=m(K^2)/m(G)$ in transformed graph after different reduction variants. Since {\pack} does not apply reductions, this column represents the square graphs $G^2$. For detailed results, see Table~\ref{tab:reduction1}.}\label{tab:summary_graph_table_nm}
\begin{tabular}{lcccccc}
           \multicolumn{1}{c}{} & 
           \multicolumn{2}{c}{{\pack}}  &
           \multicolumn{2}{c}{{\textsc{core}}} &
           \multicolumn{2}{c}{{\textsc{elaborated}}}  \\
           \cmidrule(r){2-3} \cmidrule(r){4-5} \cmidrule{6-7} 
           \multicolumn{1}{l}{\thead[l]{}} &
           \multicolumn{1}{r}{$\tilde{n}$[\%]} & \multicolumn{1}{r}{$\tilde{m}$[\%]} &
           \multicolumn{1}{r}{$\tilde{n}$[\%]} & \multicolumn{1}{r}{$\tilde{m}$[\%]}  &
           \multicolumn{1}{r}{$\tilde{n}$[\%]} & \multicolumn{1}{r}{$\tilde{m}$[\%]}  \\
        \toprule
        \multicolumn{1}{l}{\textit{planar}} & \mc{1}{r}{\numprint{100}} & \mc{1}{r}{\numprint{212.11}} & \mc{1}{r}{{\numprint{99.91}}} & \mc{1}{r}{{\numprint{211.67}}} & \mc{1}{r}{{\numprint{99.91}}} & \mc{1}{r}{{\numprint{211.67}}}\\ 
        \multicolumn{1}{l}{\textit{social (s)}} & \mc{1}{r}{\numprint{100}} & \mc{1}{r}{\numprint{3154.00}} & \mc{1}{r}{{\numprint{14.48}}} & \mc{1}{r}{{\numprint{70.25}}} & \mc{1}{r}{{\numprint{14.48}}} & \mc{1}{r}{{\numprint{70.25}}}\\ 
        \multicolumn{1}{l}{\textit{social (l)}} & \mc{1}{r}{\numprint{100}} & \mc{1}{r}{\numprint{4327.38}} & \mc{1}{r}{{\numprint{17.12}}} & \mc{1}{r}{{\numprint{274.56}}} & \mc{1}{r}{{\numprint{17.12}}} & \mc{1}{r}{\numprint{274.58}}\\ 
        \midrule
        \multicolumn{1}{l}{overall} & \mc{1}{r}{\numprint{100}} & \mc{1}{r}{\numprint{2564.50}} & \mc{1}{r}{{\numprint{43.84}}} & \mc{1}{r}{{\numprint{185.49}}} & \mc{1}{r}{{\numprint{43.84}}} & \mc{1}{r}{\numprint{185.50}}\\ 
        \bottomrule
\end{tabular}
\end{table}

\textbf{Overview/Competing Algorithms.}
We perform a wide range of experiments. 
First we perform experiments to investigate the influence of the data reduction rules in Section~\ref{sec:experimentsReduction}. Therefore, we define three configurations for our reductions. The first variant is called {\pack} and does not include any of our proposed reductions. 
Then, in {\redpdc} we only use the Domination and Clique Reduction. For the last variant {\redpfull} we apply the full set of all special case reductions as well as the Domination and Clique Reduction. The order for reductions in {\redpdc} is [\ref{red:clique}, \ref{red:domination}] while {\redpfull} uses [\ref{red:deg0}, \ref{red:deg0triangle}, \ref{red:deg1}, \ref{red:triangle}, \ref{red:cycle}, \ref{red:v-shape}, \ref{red:twin}, \ref{red:fastDomination}, \ref{red:domination}, \ref{red:clique}]. Note that we did not experiment with different orderings for the reductions since Großmann~\etal show in~\cite{DBLP:conf/gecco/GrossmannL0S23,grossmann2022heuristic} that for the MWIS problem an intuitive ordering worked best and small changes do not effect the solution quality and running time by a lot.
We then compare our algorithms against the state-of-the-art for the problem in Section~\ref{sec:experimentsSOA}. 
In particular, we compare against the genetic algorithm \textsc{gen2pack} by Trejo-Sánchez~\etal\cite{trejo2020genetic} as well as the \textsc{Apx-2p + Imp2p} algorithm by Trejo-Sánchez \etal\cite{trejo2023approximation} which only works for planar graphs. We use two configurations of \textsc{Apx-2p + Imp2p}. The configurations differ in the parameter $h$ which specifies, the number of vertices in the subgraphs. With increasing $h$ the solution quality improves, but the slower the algorithm performs. We chose the default configuration with $h=50$ and $h=100$ to improve the solution and give a fairer comparison with our 10 hour time limit.
We could not perform experiments with \textsc{Maximum-2-Pack-Cactus} since the code is not available~\cite{personalCommunicationLamas} and the data in the paper itself is presented such that a direct comparison is not possible. 

\textbf{Data Sets.}
We collected a wide range of instances for our experiments different sources.
First, we use a set of social networks in our benchmark which are typically used to benchmark recent independent set algorithms. More precisely, we use forty social networks from~\cite{nr} and~\cite{DBLP:reference/snam/BaderKM00W18}. We split this class up into two sets of twenty graphs each by the number of vertices in the graphs, i.e.~we put graphs with less than \numprint{50000} vertices in \textit{small social} subset and the others in \textit{large social}. We only added instances for the benchmark were the number of edges did not exceed 32-bit when constructing the two-neighborhood with the {\pack} algorithm. Note that this only happened for large graphs and not for any of the instances that were used in previous experiments for the maximum 2-packing set~problem.
When comparing against the other algorithms, we also include 
 \numprint{20} cactus graphs from~\cite{flores2018algorithm} as well as a random selection of \numprint{1050}~Erdös–Rényi graphs from~\cite{trejo2020genetic}. %
To be able to compare against \textsc{Apx-2p + Imp2p}, we also include planar~graphs~from~\cite{trejo2023approximation}. For an overview of the graph properties we refer to~\cite{DBLP:journals/corr/abs-2308-15515}.\newline

\vspace*{-.75cm}
\subsection{Impact of Data Reductions}\label{sec:experimentsReduction}
We now investigate the effectiveness of the data reductions. To do so, we use the three reduction configurations for the algorithm {\epack}. To evaluate the effectiveness we do not investigate the influence on \textit{Erdös–Rényi} as well as \textit{Cactus} graphs, as they are already very small. 
We compare the impact on the size of the square kernels as well as solution quality and running time. Details are presented in Tables~\ref{tab:reduction1} and~\ref{tab:reduction2}.  
We summarize these results in Table~\ref{tab:summary_graph_table_nm} and Figure~\ref{fig:performance_variants}.

First, we look at the effectiveness of our reductions on the size of the square kernels $K^2$.
When applying the graph transformation on the original input, i.e.~without applying any reduction ({\pack}), the resulting instance has the same amount of vertices and on average \numprint{25.65}$\cdot m(G)$ edges, whereas {\redpdc} and {\redpfull} both yield on average $0.44\cdot|V|$ vertices and $1.86\cdot|E|$ edges. Thus, our data reductions help to decrease the size of the transformed graph by more than a factor of ten on average.
The approaches {\redpdc} and {\redpfull} compute overall the same kernel sizes.
However, on five large social instances {\redpfull} computes smaller kernels, but only with a very small difference. On planar graphs our reduction rules, as expected, are not working good since they have mesh-like structure. The number of vertices in the kernel is only reduced by \numprint{0.1}\% compared to the original graph. All together we are able to reduce 15 out of 60 instances to an empty kernel and thereby solve them solely by our data-reductions. 
Hence, we conclude that the reductions are highly effective in reducing the graph size and especially reduce the size for social networks. 

\begin{figure}[t]
	\centering
		\caption{Solution quality (left) and running time (right) for {\epack} on \textit{planar} and \textit{social} graphs comparing our different reduction variants.  \label{fig:performance_variants}}
        \vspace*{-.5cm}
	\includegraphics[width=.48\textwidth]{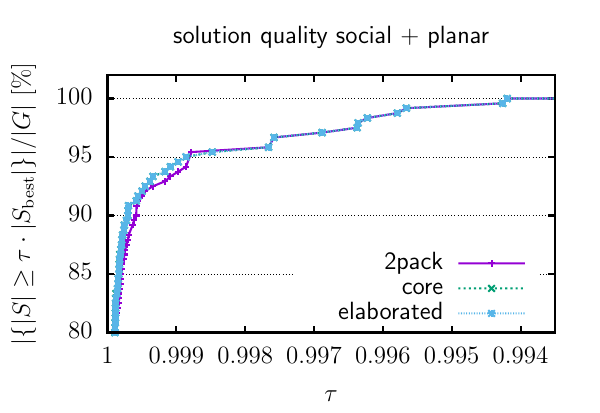}
	\includegraphics[width=.48\textwidth]{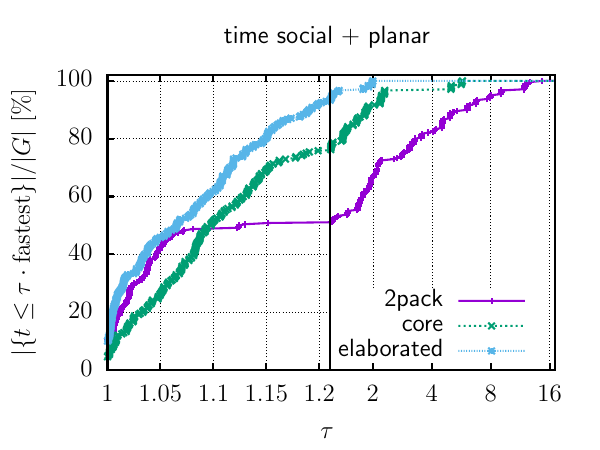}%
        \vspace*{-.5cm}
\end{figure}

Regarding solution quality overall our variant {\redpfull} is performing only slightly better compared to the other variants. Especially on the \textit{small social} and \textit{planar} graphs hardly a difference can be seen. When considering \textit{large social} graphs, however, we can find no instance, on which {\pack} outperforms {\redpdc} or {\redpfull}. Overall,~we achieve an improvement through {\redpfull} on this graph class of \numprint{0.05}\% compared to {\pack} and {\redpdc}. The instance with the largest difference in solution quality is \textit{road\_usa}. Here, {\redpfull} achieves an improvement of \numprint{0.93}\% over the other two strategies. On the instance \textit{amazon-2008} {\redpfull} performs worse compared to {\redpdc}. On this, the solution quality of {\redpdc} is improved by \numprint{0.10}\% compared to the solution of {\redpfull}. For all of the 6 instances, on which {\redpfull} was outperformed by {\pack} the improvement over {\redpfull} is always smaller than \numprint{0.01}\%.

Figure~\ref{fig:performance_variants} also show that using our different data reduction rules as a preprocessing step ({\redpdc} and {\redpfull}) especially improves the running time compared to {\pack}. Here, we see that in general our reductions are improving the performance and our approach {\redpfull} works best. In the detailed results in Table~\ref{tab:reduction2} we see that especially for \textit{large social} graphs {\redpfull} yields a speed up of \numprint{2.7} compared to {\pack}. On \textit{planar} graphs, where our reductions are not effective in reducing the initial input size, the performance is very similar \hbox{for all our variants.}

\textit{Conclusion:} We conclude that for all examined criteria, which are the size of the square kernel as well as solution quality and running time, our variant {\redpfull} performs best. Hence, we choose the {\redpfull} reduction variant in the following state-of-the-art comparisons for both {\epack} and {\hpack}.

\subsection{Comparison against State-of-the-Art}\label{sec:experimentsSOA}
We now compare our algorithms {\epack} and {\hpack} using our best reduction variant {\redpfull}, against other state-of-the-art algorithms which are \textsc{gen2pack} by Trejo-Sánchez~\etal\cite{trejo2020genetic} as well as two configurations of \textsc{Apx-2p + Imp2p} by Trejo-Sánchez~\etal\cite{trejo2023approximation}.

\textit{\textsc{gen2pack}:} 
For the comparison against \textsc{gen2pack} we only use \textit{cactus} graphs and Erdös–Rényi (\textit{erdos}) networks, \ie the instances used in their paper, as \textsc{gen2pack} is not able to solve any of the other, larger graphs within the given time limit. 
This can be explained by the initial computations containing matrix multiplication used in \textsc{gen2pack}. This does not finish during the 10 hours limit, so that the algorithm could not compute any solution at all.
 Detailed per instance results for this comparison can be found in~Table~\ref{tab:soa_small}. 
\begin{table*}
	\centering
	\caption{Summary comparison of state-of-the-art. Detailed results are presented in Tables~\ref{tab:soa_small} to \ref{tab:social_heuristic}. Geometric mean over different graph classes of solution size $|S|$ and time $t$ to find it. \textbf{Best} results are emphasized in bold. \textsc{Apx-2P + Im2p} with $h=100$ is omitted as not all planar instances where solved within the time limit.}\label{tab:soa_overal}		
	\hspace*{-.5cm}	\resizebox{.9\linewidth}{!}{\begin{minipage}{1.03\linewidth}
	\begin{tabular}{lcccccccc} 
\mc{3}{l}{} & \mc{2}{c}{\textsc{Apx-2p + Im2p}}   & \mc{2}{c}{{\redp}}  & \mc{2}{c}{{\redp}} \\
\mc{1}{l}{} & \mc{2}{c}{\textsc{gen2pack}} & \mc{2}{c}{\textsc{($h=50$)}}  & \mc{2}{c}{\textsc{b\&r}}  & \mc{2}{c}{ \textsc{heuristic}} \\ 
\cmidrule(r){2-3}\cmidrule(r){4-5}\cmidrule(r){6-7}\cmidrule(r){8-9}
\mc{1}{l}{\thead[l]{Class}} &  \mc{1}{c}{\thead[c]{$|S|$}} & \mc{1}{c}{\thead[c]{$t$ [s]}} & \mc{1}{c}{\thead[c]{$|S|$}} & \mc{1}{c}{\thead[c]{$t$ [s]}} &  \mc{1}{c}{\thead[c]{$|S|$}} & \mc{1}{c}{\thead[c]{$t$ [s]}} & \mc{1}{c}{\thead[c]{$|S|$}} & \mc{1}{c}{\thead[c]{$t$ [s]}}\\ 
\toprule
\mc{1}{l}{\textit{cactus}} & \mc{1}{r}{\numprint{104}} & \mc{1}{r}{\numprint{1384007.11}} &  \mc{2}{c}{-} & \mc{1}{r}{\textbf{\numprint{137}}} & \mc{1}{r}{\textbf{\numprint{4.26}}} & \mc{1}{r}{\textbf{\numprint{137}}} & \mc{1}{r}{\numprint{7.30}}\\
\mc{1}{l}{\textit{erdos}} & \mc{1}{r}{\numprint{8}} & \mc{1}{r}{\numprint{21679.67}} &   \mc{2}{c}{-}& \mc{1}{r}{\textbf{\numprint{9}}} & \mc{1}{r}{\textbf{\numprint{0.31}}} & \mc{1}{r}{\textbf{\numprint{9}}} & \mc{1}{r}{\numprint{0.53}}\\  
\mc{1}{l}{\textit{planar}}  &  \mc{2}{c}{-} & \mc{1}{r}{\numprint{110009}} & \mc{1}{r}{\numprint{255.04}} & \mc{1}{r}{\numprint{92135}} & \mc{1}{r}{\numprint{10.41}} & \mc{1}{r}{\textbf{\numprint{110095}}} & \mc{1}{r}{\textbf{\numprint{31706.65}}}\\ 
\mc{1}{l}{\textit{social (s)}} &  \mc{2}{c}{-}&  \mc{2}{c}{-}& \mc{1}{r}{\textbf{\numprint{159}}} & \mc{1}{r}{\textbf{\numprint{11.32}}} &  \mc{1}{r}{\textbf{\numprint{159}}} & \mc{1}{r}{\numprint{13.53}}\\ 
\mc{1}{l}{\textit{social (l)}} &  \mc{2}{c}{-}& \mc{2}{c}{-}&  \mc{1}{r}{\numprint{30066}} & \mc{1}{r}{\numprint{6442.49}} &  \mc{1}{r}{\textbf{\numprint{30756}}} & \mc{1}{r}{\textbf{\numprint{26377.56}}}
\end{tabular} 

	\end{minipage}}
        \vspace*{-.75cm}
\end{table*} 

In Figure~\ref{fig:performance_soa} we give performance profiles for running time and solution quality. In Table~\ref{tab:soa_overal} we give the geometric mean running times and solution qualities for these results. Our algorithm {\epack} as well as {\hpack} \emph{find overall the optimal solution in the classes \textit{cactus} and \textit{erdos} within a few milliseconds}. 
Our algorithms dominate \textsc{gen2pack} in terms of both solution quality as well as running time. Especially the differences in running time are very large. On all graphs our two algorithms are multiple orders of magnitude faster than \textsc{gen2pack}. It can only find optimal solutions for \numprint{6} out of these~\numprint{40}~graphs, see Table~\ref{tab:soa_small}.
On these two graph classes \textit{both} our algorithms always compute the optimum solution quality which results in an average solution quality improvement of more than 20\% and a speedup of more than $10^5$. Among all instances under consideration, on \textit{Erdos37-2} \textsc{gen2pack} needed the least amount of time to compute an optimum solution. For this instance we achieve with {\epack} a speed up of more than \numprint{300000} and more than \numprint{350000} with {\hpack}. The instance on which \textsc{gen2pack} needs the most time is \textit{cac1000}. On this instance, {\epack} and {\hpack} again have similar speedups in the range of $10^5$ over \textsc{gen2pack} and an improvement in solution quality of roughly 32\%. 
Considering the overall data set, our approach {\epack} can solve 63 out of 100 graphs to optimality within less than one second and 71 within the 10 hour time limit and 100GB restriction. For instances that we could not solve to optimality due to experimental restrictions, we give the solution found until this point, see Tables~\ref{tab:reduction2}~and~\ref{tab:soa_small}. 

In Table~\ref{tab:social_heuristic} we compare {\epack} and {\hpack} on social graphs, which are not solvable with the competitors. On these instances we are able to achieve an average improvement in solution quality of around 1\% with {\hpack} compared to {\epack}. Especially for large graphs, where our exact solver meets the memory threshold, our heuristic variant is able to outperform the results of {\epack}. %

\textit{\textsc{Apx-2p + Imp2p}}:  
Since our reductions do not perform well on planar graphs, we are not able to solve them to optimality with {\epack} and exceed the memory threshold quite fast. 
 Overall, the solution quality we achieve with {\epack} for the planar graphs is \numprint{84}\% of the solution quality that the competitor \textsc{Apx-2p + Imp2p} ($h=50$) computes, but we only use \numprint{4}\% of its time needed. 
 With {\hpack}, on the other hand, we outperform \textsc{Apx-2p + Imp2p} ($h=50$) on all but one instance regarding solution quality, see Table~\ref{tab:soa_small}. We achieve an average solution quality which is on par, i.e.~our improvement over the competitors results by \numprint{0.08}\%. Note that the authors experimentally show in~\cite{trejo2023approximation} that the computed 2-packing set by \textsc{Apx-2p + Imp2p} is already at least \numprint{99}\% of the optimum solution. %
 However, on those instances our algorithms needs roughly two orders of magnitude more running time than \textsc{Apx-2p + Imp2p} ($h=50$). 
 \textsc{Apx-2p + Imp2p} with $h=100$ compared to $h=50$ can improve all but one solution, however the running time increase is up to multiple orders of magnitude and some instances where not solved within the 10 hour time limit. {\hpack} is able to find better solutions than \textsc{Apx-2p + Imp2p} ($h=100$) on 9 out of 20 instances. Again the differences in solution quality are very small. 
 Detailed results for these experiments are given in Table~\ref{tab:soa_small} and Figure~\ref{fig:performance_soa} (right). The competitor \textsc{gen2pack} for general graphs is unable to solve any of these instances which is why we omitted it in the corresponding tables and performance profiles.
\begin{figure*}[t]
	\centering
	\caption{State of the art comparison on solution quality (top) and running time (bottom) for different graph classes. For planar graphs (right) the competitor \textsc{gen2pack} is not able to solve the instances. To this end we add the algorithm \textsc{Apx-2p + Imp2p} with $h=50$ and $h=100$ as well as {\hpack} for comparison.} \label{fig:performance_soa}
	\includegraphics[width=.325\textwidth]{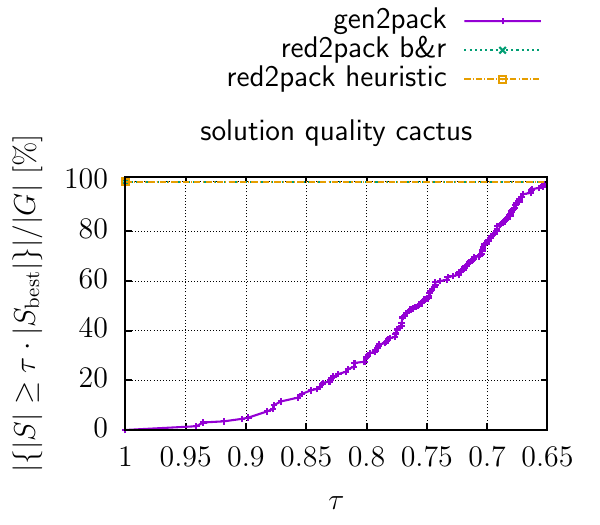}
	\includegraphics[width=.325\textwidth]{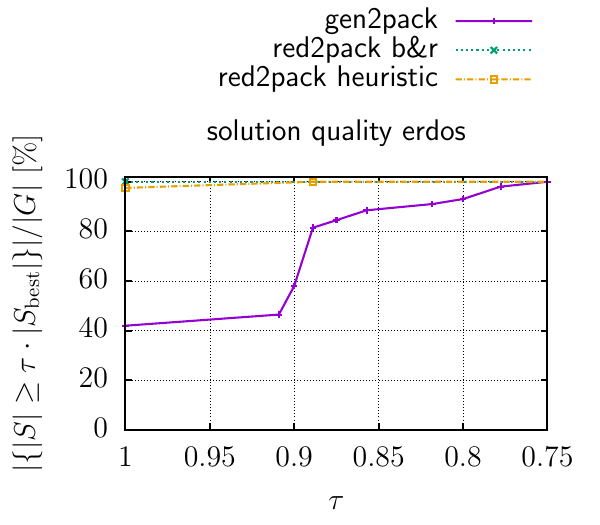}
	\includegraphics[width=.325\textwidth]{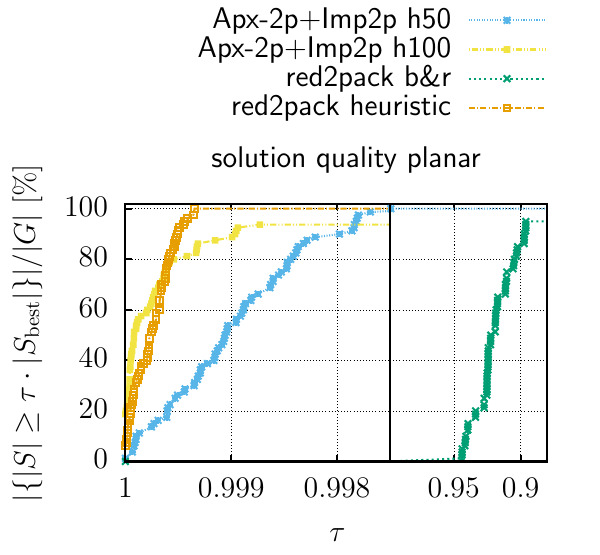}
\vspace*{-.2cm}
	\includegraphics[width=.325\textwidth]{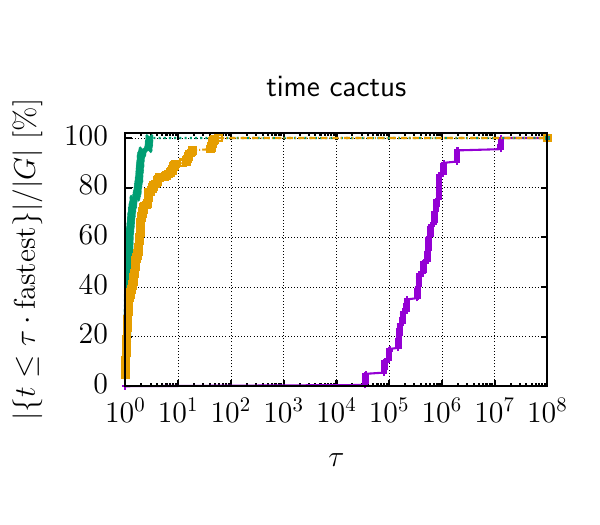}
	\includegraphics[width=.325\textwidth]{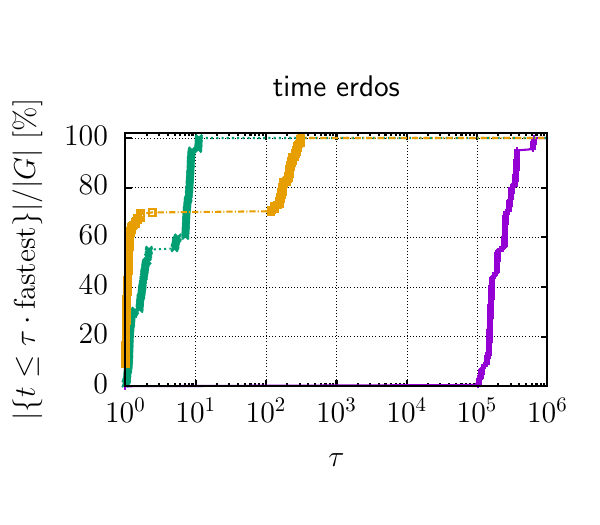}
	\includegraphics[width=.325\textwidth]{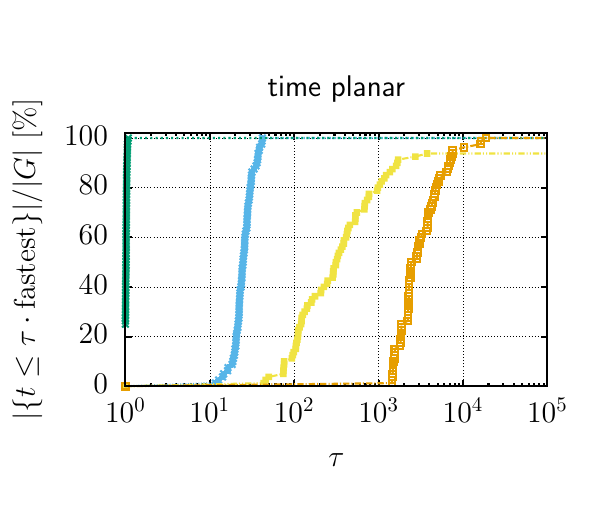}
        \vspace*{-.5cm}
\end{figure*}

\textit{Conclusion:} We conclude that on all instances we are able to outperform the competitor \textsc{gen2pack} for arbitrary graphs in both solution quality and running time by multiple orders of magnitude. Moreover, we can solve a wide range of instances to optimality that previously have been unsolvable.
When comparing against algorithms specialized on planar graphs, we presented two options: one that is by more than a factor of \numprint{24} times faster with lower solution quality and one that is on par in terms of solution quality, but slower, compared to both configurations of the state-of-the-art \emph{specialized} solver for planar graphs.

\section{Conclusion and Future Work} \label{sec:conclusion}
This work introduces novel data reduction rules to solve the maximum 2-packing set problem as well as proposes a new exact algorithm {\epack} that uses these reductions to exactly solve the maximum 2-packing set problem on large-scale arbitrary graphs. Additionally, a new heuristic algorithm {\hpack} is introduced that works similar.
Both of the algorithms {\epack} and {\hpack} work in three phases. First the new data reduction rules are applied to the given input
resulting in a reduced instance. Following the reduction phase, the resulting graph is transformed,
such that a solution on the transformed graph for the maximum independent set problem corresponds to a solution of
the maximum 2-packing set problem for the original graph. The third phase of the algorithms consists of solving the maximum independent set problem on the transformed graph.
Our tests indicate that our algorithms outperform the previous best algorithm for arbitrary graphs both in terms of solution quality and running time on all instances. For instance, we can compute optimal solutions for 63\% of our graphs in under a second, whereas the competing method for arbitrary graphs achieves this only for 5\% of the graphs even with a 10-hour time frame. Furthermore, our method successfully solves many large instances that remained unsolved before. Lastly, our algorithm can compete with a specialized solver on planar instances in terms of solution size and computes near optimum solutions. 
Our code is publicly available under \url{https://github.com/KarlsruheMIS/red2pack}.

In future work, we want to find more reduction rules, especially for mesh-like graphs. We are also interested in the weighted 2-packing set problem as well as the $k$-packing set problem for larger values~of~$k$ and find independent motifs in graphs via hypergraphs \hbox{matching algorithms.}
\vfill \pagebreak

\bibliographystyle{plain}
\bibliography{phdthesiscs, mwis, red2pack, all_clean}
\vfill \pagebreak

\begin{appendix}

\section{Reduction and Transformation Details}

	\begin{table}[H]
				\caption{Percentage of remaining vertices and edges for different reduction variants after transformation. Bolt numbers indicate the \textbf{best} results, gray background marks \colorbox{lightgray}{empty kernels}.}\label{tab:reduction1}
		\resizebox{\linewidth}{!}{\begin{minipage}{1.55\linewidth}		
				\centering
				\vspace*{-.15cm}
\begin{tabular}{llcccccc} 

\mc{2}{l}{}  & \mc{2}{c}{}  & \mc{2}{c}{\epack}  & \mc{2}{c}{\epack} \\  
\mc{2}{l}{}  & \mc{2}{c}{\pack}  & \mc{2}{c}{\textsc{core}}  & \mc{2}{c}{\textsc{elaborated}} \\  
\cmidrule(r){3-4}\cmidrule(r){5-6}\cmidrule(r){7-8}
\mc{1}{l}{\thead[l]{Class}} & \mc{1}{l}{\thead[l]{Graphs}} & \mc{1}{c}{\thead[c]{$n(K^2)/n~[\%]$}} & \mc{1}{c}{\thead[c]{$m(K^2)/m~[\%]$}} & \mc{1}{c}{\thead[c]{$n(K^2)/n~[\%]$}} & \mc{1}{c}{\thead[c]{$m(K^2)/m~[\%]$}} & \mc{1}{c}{\thead[c]{$n(K^2)/n~[\%]$}} & \mc{1}{c}{\thead[c]{$m(K^2)/m~[\%]$}}\\ 
 \midrule 
\rowcolor{lightgray} \mc{1}{l}{\cellcolor{white} \textit{social (s)}} & \mc{1}{l}{\textit{PGPgiantcompo}} & \mc{1}{r}{\numprint{100.00}} & \mc{1}{r}{\numprint{873.91}} & \mc{1}{r}{\textbf{\numprint{0.00}}} & \mc{1}{r}{\textbf{\numprint{0.00}}} & \mc{1}{r}{\textbf{\numprint{0.00}}} & \mc{1}{r}{\textbf{\numprint{0.00}}} \\ 
\rowcolor{lightgray} \mc{1}{l}{\cellcolor{white} } & \mc{1}{l}{\textit{adjnoun}} & \mc{1}{r}{\numprint{100.00}} & \mc{1}{r}{\numprint{725.18}} & \mc{1}{r}{\textbf{\numprint{0.00}}} & \mc{1}{r}{\textbf{\numprint{0.00}}} & \mc{1}{r}{\textbf{\numprint{0.00}}} & \mc{1}{r}{\textbf{\numprint{0.00}}} \\ 
\rowcolor{lightgray} \mc{1}{l}{\cellcolor{white}} & \mc{1}{l}{\textit{as-22july06}} & \mc{1}{r}{\numprint{100.00}} & \mc{1}{r}{\numprint{22941.92}} & \mc{1}{r}{\textbf{\numprint{0.00}}} & \mc{1}{r}{\textbf{\numprint{0.00}}} & \mc{1}{r}{\textbf{\numprint{0.00}}} & \mc{1}{r}{\textbf{\numprint{0.00}}} \\ 
\rowcolor{lightgray} \mc{1}{l}{\cellcolor{white}} & \mc{1}{l}{\textit{celegans\textunderscore metabolic}} & \mc{1}{r}{\numprint{100.00}} & \mc{1}{r}{\numprint{2239.56}} & \mc{1}{r}{\textbf{\numprint{0.00}}} & \mc{1}{r}{\textbf{\numprint{0.00}}} & \mc{1}{r}{\textbf{\numprint{0.00}}} & \mc{1}{r}{\textbf{\numprint{0.00}}} \\
\bottomrule 

\end{tabular} 

		\end{minipage}}
	\end{table}

\begin{table}[H]
	\caption{Solution size $|S|$ and time $t$ needed to find it. The time $t_p$ is the time needed to prove the optimality. Bolt numbers indicate \textbf{best} results, gray background \colorbox{lightgray}{optimally} solved instances.}\label{tab:reduction2}
		\resizebox{\linewidth}{!}{\begin{minipage}{1.525\linewidth}
			\centering
			\begin{tabular}{llccccccccc} 
	\mc{2}{l}{}  & \mc{3}{c}{}  & \mc{3}{c}{\epack}  & \mc{3}{c}{\epack} \\  
	\mc{2}{l}{}  & \mc{3}{c}{\pack}  & \mc{3}{c}{\textsc{core}}  & \mc{3}{c}{\textsc{elaborated}} \\  
\cmidrule(r){3-5}\cmidrule(r){6-8}\cmidrule(r){9-11}
\mc{1}{l}{\thead[l]{Class}} & \mc{1}{l}{\thead[l]{Graphs}} & \mc{1}{c}{\thead[c]{$|S|$}} & \mc{1}{c}{\thead[c]{$t$ [s]}} & \mc{1}{c}{\thead[c]{$t_p$ [s]}} & \mc{1}{c}{\thead[c]{$|S|$}} & \mc{1}{c}{\thead[c]{$t$ [s]}} & \mc{1}{c}{\thead[c]{$t_p$ [s]}} & \mc{1}{c}{\thead[c]{$|S|$}} & \mc{1}{c}{\thead[c]{$t$ [s]}} & \mc{1}{c}{\thead[c]{$t_p$ [s]}}\\ 
 \midrule 
\rowcolor{lightgray} \mc{1}{l}{\cellcolor{white} \textit{social (s)}} & \mc{1}{l}{\textit{PGPgiantcompo}} & \mc{1}{r}{\textbf{\numprint{2708}}} & \mc{1}{r}{\numprint{0.02}} & \mc{1}{r}{\numprint{0.02}} & \mc{1}{r}{\textbf{\numprint{2708}}} & \mc{1}{r}{\numprint{0.01}} & \mc{1}{r}{\numprint{0.01}} & \mc{1}{r}{\textbf{\numprint{2708}}} & \mc{1}{r}{\textbf{\numprint{0.01}}} & \mc{1}{r}{\textbf{\numprint{0.01}}} \\ 
\rowcolor{lightgray} \mc{1}{l}{\cellcolor{white}} & \mc{1}{l}{\textit{adjnoun}} & \mc{1}{r}{\textbf{\numprint{18}}} & \mc{1}{r}{<\numprint{0.01}} & \mc{1}{r}{<\numprint{0.01}} & \mc{1}{r}{\textbf{\numprint{18}}} & \mc{1}{r}{\textbf{<\numprint{0.01}}} & \mc{1}{r}{\textbf{<\numprint{0.01}}} & \mc{1}{r}{\textbf{\numprint{18}}} & \mc{1}{r}{<\numprint{0.01}} & \mc{1}{r}{<\numprint{0.01}} \\ 
\rowcolor{lightgray} \mc{1}{l}{\cellcolor{white}} & \mc{1}{l}{\textit{as-22july06}} & \mc{1}{r}{\textbf{\numprint{2026}}} & \mc{1}{r}{\numprint{4.26}} & \mc{1}{r}{\numprint{4.26}} & \mc{1}{r}{\textbf{\numprint{2026}}} & \mc{1}{r}{\textbf{\numprint{0.95}}} & \mc{1}{r}{\textbf{\numprint{0.96}}} & \mc{1}{r}{\textbf{\numprint{2026}}} & \mc{1}{r}{\numprint{1.74}} & \mc{1}{r}{\numprint{1.74}} \\
\bottomrule 

\end{tabular} 

		\end{minipage}}
\end{table}

\section{Detailed State-of-the-Art Comparison}
 \begin{table}[h!]
 	\caption{Solution size $|S|$ and time $t$ needed to find it. The time $t_p$ is the time needed to prove the optimality. Bolt numbers indicate \textbf{best} results in comparison. Gray background marks \colorbox{lightgray}{optimally} solved instances. No competitor is able to solve these instances.}\label{tab:social_heuristic}
		\resizebox{\linewidth}{!}{\begin{minipage}{1.525\linewidth}
			\centering
			\begin{tabular}{lcccccc} 
\mc{2}{l}{}  & \mc{3}{c}{{\redp}}  & \mc{2}{c}{{\redp}} \\ 
\mc{2}{l}{}  & \mc{3}{c}{ \textsc{b\&r}}  & \mc{2}{c}{\textsc{heuristic}} \\ 
\cmidrule(r){3-5}\cmidrule(r){6-7}
\mc{1}{l}{\thead[l]{Class}} & \mc{1}{l}{\thead[l]{Graphs}} & \mc{1}{c}{\thead[c]{$|S|$}} & \mc{1}{c}{\thead[c]{$t$ [ms]}} & \mc{1}{c}{\thead[c]{$t_p$ [ms]}} & \mc{1}{c}{\thead[c]{$|S|$}} & \mc{1}{c}{\thead[c]{$t$ [ms]}}\\ 
 \midrule 
\rowcolor{lightgray} \mc{1}{l}{\cellcolor{white} \textit{social (s)}} &\mc{1}{l}{\textit{PGPgiantcompo}} & \mc{1}{r}{\textbf{\numprint{2708}}} & \mc{1}{r}{\numprint{8.34}} & \mc{1}{r}{\textbf{\numprint{8.74}}} & \mc{1}{r}{\textbf{\numprint{2708}}} & \mc{1}{r}{\textbf{\numprint{8.31}}} \\ 
\rowcolor{lightgray} \mc{1}{l}{\cellcolor{white}} & \mc{1}{l}{\textit{adjnoun}} & \mc{1}{r}{\textbf{\numprint{18}}} & \mc{1}{r}{\textbf{\numprint{0.43}}} & \mc{1}{r}{\textbf{\numprint{0.48}}} & \mc{1}{r}{\textbf{\numprint{18}}} & \mc{1}{r}{\textbf{\numprint{0.43}}} \\ 
\rowcolor{lightgray} \mc{1}{l}{\cellcolor{white}} & \mc{1}{l}{\textit{as-22july06}} & \mc{1}{r}{\textbf{\numprint{2026}}} & \mc{1}{r}{\textbf{\numprint{1737.58}}} & \mc{1}{r}{\textbf{\numprint{1739.08}}} & \mc{1}{r}{\textbf{\numprint{2026}}} & \mc{1}{r}{\numprint{1774.51}} \\ 
\rowcolor{lightgray} \mc{1}{l}{\cellcolor{white}} & \mc{1}{l}{\textit{celegans\textunderscore metabolic}} & \mc{1}{r}{\textbf{\numprint{29}}} & \mc{1}{r}{\textbf{\numprint{4.77}}} & \mc{1}{r}{\textbf{\numprint{4.84}}} & \mc{1}{r}{\textbf{\numprint{29}}} & \mc{1}{r}{\numprint{4.81}} \\ 
\rowcolor{lightgray} \mc{1}{l}{\cellcolor{white}} & \mc{1}{l}{\textit{celegansneural}} & \mc{1}{r}{\textbf{\numprint{14}}} & \mc{1}{r}{\textbf{\numprint{55.94}}} & \mc{1}{r}{\textbf{\numprint{58.29}}} & \mc{1}{r}{\textbf{\numprint{14}}} & \mc{1}{r}{\numprint{57.74}} \\
\bottomrule 

\end{tabular} 

		\end{minipage}}
\end{table} 

\begin{table}[h!]
	\caption{Solution size $|S|$ and time $t$ needed to find it. The time $t_p$ is the time needed to prove the optimality. Bolt numbers indicate \textbf{best} results in comparison. Gray background marks \colorbox{lightgray}{optimally} solved instances.}\label{tab:soa_small}
	\resizebox{\linewidth}{!}{\begin{minipage}{1.525\linewidth}
			\centering
			\begin{tabular}{llccccccc} 
\mc{2}{l}{}  & \mc{2}{c}{}  & \mc{3}{c}{{\redp}}  & \mc{2}{c}{{\redp}} \\ 
\mc{2}{l}{}  & \mc{2}{c}{\textsc{gen2pack}}  & \mc{3}{c}{ \textsc{b\&r}}  & \mc{2}{c}{\textsc{heuristic}} \\ 
\cmidrule(r){3-4}\cmidrule(r){5-7}\cmidrule(r){8-9}
\mc{1}{l}{\thead[l]{Class}} & \mc{1}{l}{\thead[l]{Graphs}} & \mc{1}{c}{\thead[c]{$|S|$}} & \mc{1}{c}{\thead[c]{$t$ [ms]}} & \mc{1}{c}{\thead[c]{$|S|$}} & \mc{1}{c}{\thead[c]{$t$ [ms]}} & \mc{1}{c}{\thead[c]{$t_p$ [ms]}} & \mc{1}{c}{\thead[c]{$|S|$}} & \mc{1}{c}{\thead[c]{$t$ [ms]}}\\ 
 \midrule 
\rowcolor{lightgray} \mc{1}{l}{\cellcolor{white} \textit{cactus}} & \mc{1}{l}{\textit{cac50}} & \mc{1}{r}{\numprint{15}} & \mc{1}{r}{\numprint{25291.12}} & \mc{1}{r}{\textbf{\numprint{17}}} & \mc{1}{r}{\textbf{\numprint{0.05}}} & \mc{1}{r}{\textbf{\numprint{0.09}}} & \mc{1}{r}{\textbf{\numprint{17}}} & \mc{1}{r}{\textbf{\numprint{0.05}}} \\ 
\rowcolor{lightgray} \mc{1}{l}{\cellcolor{white}} & \mc{1}{l}{\textit{cac100}} & \mc{1}{r}{\numprint{28}} & \mc{1}{r}{\numprint{85674.26}} & \mc{1}{r}{\textbf{\numprint{31}}} & \mc{1}{r}{\textbf{\numprint{0.09}}} & \mc{1}{r}{\textbf{\numprint{0.14}}} & \mc{1}{r}{\textbf{\numprint{31}}} & \mc{1}{r}{\textbf{\numprint{0.09}}} \\ 
\rowcolor{lightgray} \mc{1}{l}{\cellcolor{white}} & \mc{1}{l}{\textit{cac150}} & \mc{1}{r}{\numprint{42}} & \mc{1}{r}{\numprint{181973.07}} & \mc{1}{r}{\textbf{\numprint{49}}} & \mc{1}{r}{\textbf{\numprint{0.31}}} & \mc{1}{r}{\textbf{\numprint{0.31}}} & \mc{1}{r}{\textbf{\numprint{49}}} & \mc{1}{r}{\numprint{13.49}} \\ 
\rowcolor{lightgray} \mc{1}{l}{\cellcolor{white}} & \mc{1}{l}{\textit{cac200}} & \mc{1}{r}{\numprint{55}} & \mc{1}{r}{\numprint{313878.10}} & \mc{1}{r}{\textbf{\numprint{65}}} & \mc{1}{r}{\textbf{\numprint{10.27}}} & \mc{1}{r}{\textbf{\numprint{10.70}}} & \mc{1}{r}{\textbf{\numprint{65}}} & \mc{1}{r}{\numprint{14.47}} \\
\bottomrule

\end{tabular} 
				
			\vspace*{-.25cm}
			\begin{tabular}{llcccccccc} 
\mc{2}{l}{}  & \mc{2}{c}{\textsc{Apx-2p + Im2p}}  & \mc{2}{c}{\textsc{Apx-2p + Im2p}}  & \mc{2}{c}{{\redp}}  & \mc{2}{c}{{\redp} } \\
 \mc{2}{l}{}  & \mc{2}{c}{\textsc{($h=100$)}}  & \mc{2}{c}{\textsc{($h=50$)}}  & \mc{2}{c}{\textsc{b\&r}}  & \mc{2}{c}{ \textsc{heuristic}} \\ 
\cmidrule(r){3-4}\cmidrule(r){5-6}\cmidrule(r){7-8}\cmidrule(r){9-10}
\mc{1}{l}{\thead[l]{Class}} & \mc{1}{l}{\thead[l]{Graphs}} & \mc{1}{c}{\thead[c]{$|S|$}} & \mc{1}{c}{\thead[c]{$t$ [s]}} & \mc{1}{c}{\thead[c]{$|S|$}} & \mc{1}{c}{\thead[c]{$t$ [s]}} & \mc{1}{c}{\thead[c]{$|S|$}} & \mc{1}{c}{\thead[c]{$t$ [s]}} & \mc{1}{c}{\thead[c]{$|S|$}} & \mc{1}{c}{\thead[c]{$t$ [s]}}\\ 
 \midrule 
\mc{1}{l}{\cellcolor{white} \textit{planar}} & \mc{1}{l}{\textit{outP500\textunderscore 1}} & \mc{1}{r}{\textbf{\numprint{20333}}} & \mc{1}{r}{\textbf{\numprint{97.51}}} & \mc{1}{r}{\numprint{20327}} & \mc{1}{r}{\numprint{22.99}} & \mc{1}{r}{\numprint{19140}} & \mc{1}{r}{\numprint{1.77}} & \mc{1}{r}{\numprint{20332}} & \mc{1}{r}{\numprint{26166.03}} \\ 
\mc{1}{l}{\cellcolor{white}} & \mc{1}{l}{\textit{outP500\textunderscore 2}} & \mc{1}{r}{\textbf{\numprint{24187}}} & \mc{1}{r}{\textbf{\numprint{1978.02}}} & \mc{1}{r}{\numprint{24144}} & \mc{1}{r}{\numprint{58.42}} & \mc{1}{r}{\numprint{21894}} & \mc{1}{r}{\numprint{2.25}} & \mc{1}{r}{\numprint{24181}} & \mc{1}{r}{\numprint{12631.72}} \\ 
\mc{1}{l}{\cellcolor{white}} & \mc{1}{l}{\textit{outP1000\textunderscore 1}} & \mc{1}{r}{\textbf{\numprint{48575}}} & \mc{1}{r}{\textbf{\numprint{1066.37}}} & \mc{1}{r}{\numprint{48472}} & \mc{1}{r}{\numprint{145.79}} & \mc{1}{r}{\numprint{43855}} & \mc{1}{r}{\numprint{4.76}} & \mc{1}{r}{\numprint{48567}} & \mc{1}{r}{\numprint{29215.74}} \\
\bottomrule

\end{tabular} 

	\end{minipage}}
\end{table}

\end{appendix}

\end{document}